\documentclass[conference]{IEEEtran}
\IEEEoverridecommandlockouts

\usepackage{cite}
\usepackage{amsmath,amssymb,amsfonts}
\usepackage{algorithmic}
\usepackage{graphicx}
\usepackage{textcomp}
\usepackage{xcolor}
\def\BibTeX{{\rm B\kern-.05em{\sc i\kern-.025em b}\kern-.08em
    T\kern-.1667em\lower.7ex\hbox{E}\kern-.125emX}}
\usepackage{amsthm}

\theoremstyle{definition}
\newtheorem{thm}{Theorem}

\theoremstyle{definition}
\newtheorem{defn}{Definition}

\theoremstyle{lemma}
\newtheorem{lemma}{Lemma}

\begin{document}

\title{Hemispherical Concentration Subset Recovery in Many-Access Gaussian Multiple-Access Channels}

\author{\IEEEauthorblockN{Nazanin Mirhosseini}
\IEEEauthorblockA{\textit{Dept. of Electrical and Computer Engineering} \\
\textit{Colorado State University}\\ Fort Collins, CO, USA \\
nazmir@colostate.edu}}

\maketitle

\begin{abstract}
	We consider subset recovery in the many-access Gaussian multiple-access channel with a shared spherical codebook, where codewords are drawn independently and uniformly from the hypersphere of radius \( \sqrt{nP} \), the number of active users scales linearly with the blocklength $n$ as \( K_a(n)=\beta n \) for a constant \( \beta > 0 \), and the codebook size is \( M_n=n^d \) with \( d>2 \). We identify a geometric property showing that, for \( 0<\beta<2 \), any transmitted \( K_a(n) \)-subset lies in a single hemisphere with high probability for sufficiently large $n$. We further show that reliable decoding is possible only for \( \beta < 1/4 \). The overlap between the reliable decoding range of \( \beta \) and the hemispherical concentration range motivates our approach of two-stage decoding procedure. In the pre-filtering stage, the decoder restricts attention to a sequence of spherical caps \( \{ \hat{\mathcal{H}}_n \} \) that converges in Hausdorff distance to the hemisphere $\hat{\mathcal{H}}$, whose axis is the normalized observation \( \hat{\mathbf{u}}=\mathbf{Y}/\|\mathbf{Y}\| \). In the second stage, maximum-likelihood decoding is performed over the reduced candidate set. We show that the per-user error probability of the pre-filtering stage vanishes as \( n\to\infty \). Moreover, the per-user error probability of the maximum-likelihood stage over the reduced search space decays exponentially with asymptotic exponent \( P/4 \).
\end{abstract}


\section{Introduction}

The growth of the Internet of Things (IoT) and massive machine-type communication (mMTC) calls for efficient resource sharing among intermittently active devices \cite{Bockelmann2016, Dawy2017}. Unsourced random access (URA) is a useful model in this context because it removes the need for identity resolution, thereby eliminating most of the coordination overhead associated with classical multiple access protocols \cite{Polyanskiy2017, Liva2024, Wu2020}. In URA, the receiver aims to recover the unordered set of transmitted messages from a shared codebook, where user identities play no role in the decoding objective \cite{Polyanskiy2017}. This formulation is particularly appealing for large-scale access, but standard URA analyses typically rely on a sparse activity regime, where the number of active users is small relative to the ambient dimension. By contrast, practical mMTC deployments are often ultra-dense. In such regimes, the active population may grow proportionally with the blocklength, rendering the standard sparsity assumption less applicable. The resulting challenge is no longer merely to identify a few active messages among many possibilities, but to distinguish a large number of simultaneously transmitted messages in a dense, many-user environment \cite{Chen2017, Kowshik2020}.
Motivated by this perspective, we study a URA-inspired subset recovery problem in a dense many-access regime. The receiver is tasked with recovering the transmitted message subset from a shared codebook, while user identities remain irrelevant. Our model retains the structural advantages of URA while allowing the number of active users to grow linearly with the blocklength. Since the codebook size grows faster (polynomially), the active-user fraction still vanishes asymptotically, but the resulting scaling is markedly different from the sparse random-access regime. Furthermore, we employ spherical codebooks, with codewords drawn uniformly from the surface of the sphere $\mathbb{S}^{n-1}(\sqrt{nP})$. Unlike i.i.d.\ Gaussian codebooks, which satisfy the power constraint only with high probability, the spherical ensemble enforces constant-energy transmission deterministically and assigns equal power to all messages \cite{Gao2025}, \cite{Shannon1959}. This geometry is natural for power-constrained multiple-access models and can sharpen the geometric separation among candidate codewords. In our setting, the spherical geometry provides a useful way to reduce the effective search space, and the per-user error probability metric remains compatible with the non-sparse subset recovery formulation under many-access scaling.

\emph{Notations}: We denote the unit hypersphere in $\mathbb{R}^{n}$ by $\mathbb{S}^{n-1}$. Convergence in probability and in distribution are denoted by $\xrightarrow{P}$ and $\xrightarrow{d}$, respectively. A hemisphere is specified by a unit vector (direction) pointing form the center toward its pole. The first two standard basis vectors are $\mbox{\boldmath $e$}_{1}=(1,0,...,0)$ and $\mbox{\boldmath $e$}_{2}=(0,1,0,...,0)$. We write $\mbox{\boldmath $s$}_{i}=\mbox{\boldmath $x$}_{i}/\sqrt{nP}$ for the normalized codeword. The binomial and Bernoulli distributions are denoted by $\mathrm{Bin}(.,.)$ and $\mathrm{Bern}(.)$, respectively. The cardinality of a set $\mathcal{A}$ is denoted by $|\mathcal{A}|$.

\section{System Model}

Each active user selects a message uniformly at random from $\mathcal{M}_n = [M_n]$, where $M_n= n^d$ for $d>2$, and the corresponding message is mapped to a codeword drawn uniformly and independently from the surface of sphere $\mathbb{S}^{n-1}(\sqrt{nP})$. The encoder for a fixed $n$ is hence $f: \mathcal{M}_n \rightarrow \mathcal{X}^{n}$, where $\mathcal{X}$ is the input alphabet. The codebook $\mathcal{C}_n = \left\{ \mbox{\boldmath $x$}_m : m \in [M_n] \right\}$ is common among users. More precisely,
\begin{equation}
	\mbox{\boldmath $x$}_m = \sqrt{nP}\,\frac{\mbox{\boldmath $X$}_m}{\|\mbox{\boldmath $X$}_m\|}, \qquad m \in [M_n],
\end{equation}
with $\mbox{\boldmath $X$}_m \sim \mathcal{N}(\mathbf{0},\mathbf{I}_n)$ independently across $m$.

We assume that the number of active users satisfies $K_a(n) = \beta n$, for some fixed $\beta > 0$. For each blocklength $n$, let $\mathcal{S} \subset [M_n]$ denote the transmitted subset of indices with $|\mathcal{S}| = K_a(n)$. The Gaussian multiple access channel output is
\begin{equation}
\mbox{\boldmath $Y$} = \sum_{m \in \mathcal{S}} \mbox{\boldmath $x$}_m + \mbox{\boldmath $Z$},
\label{Y}
\end{equation}
where noise $\mbox{\boldmath $Z$} \sim \mathcal{N}(\mathbf{0},\mathbf{I}_n)$ is independent of the codebook.

The decoder observes $\mbox{\boldmath $Y$}$ and, given the knowledge of $\beta$, aims to recover the transmitted subset. The decoder is defined for a fixed $n$ as
\begin{eqnarray}
	g : \mathcal{Y}^n \to { [M_n] \choose K_a(n)},
\end{eqnarray}
where $\mathcal{Y}$ is the output alphabet.
We describe the two-stage decoding procedure in Section \ref{GeoDec}.

\section{Hemispherical Subsets}

\begin{defn}(\textit{Hemispherical $K_a(n)$-subset}):
	For a fixed $n$, a $K_a(n)$-subset chosen from a codebook of size $M_n$ is hemispherical, if the convex hull of the codewords in the chosen subset does not contain the origin, i.e., all points in a $K_a(n)$-hemispherical subset lie only on one half of the sphere.
\end{defn}
	
\begin{thm}(\textit{Wendel's Theorem}) \cite{Wendel1962} \label{Wendel}
	The probability that $N$ points distributed uniformly at random on an $(n-1)$-sphere, all lie on some hemisphere is
	\begin{eqnarray}
		p_{n,N} = \frac{1}{2^{N-1}}\sum_{i=0}^{n-1}{N-1  \choose i}. \nonumber
	\end{eqnarray}
\end{thm}

We now adapt Wendel's theorem to our model. 
\begin{thm}\label{K_aHemispherical}
	Let $p_{n,K_a(n)}$ denote the probability that a $K_a(n)$-subset of codewords is hemispherical. Then $p_{n,K_a(n)}\to 1$ exponentially as $n\rightarrow \infty$
if and only if $1<\beta<2$. If $0<\beta \leq 1$, the probability $p_{n,K_a(n)} = 1$ for all $n$.
\end{thm}
\begin{proof}
	The sketch of the proof is the following. By Wendel's Theorem, it yields that 
	\begin{eqnarray}
		&& \!\!\!\!\!\!\!\!\!\!\!\!\! p_{n,K_a(n)} \triangleq  \mathbb{P}\left[\text{a $K_a(n)$-subset is hemispherical}\right]\nonumber\\
		&& \ \ \ \ \ = \frac{1}{2^{K_a(n) - 1}} \sum_{i=0}^{n-1}{K_a(n)-1 \choose i}. \label{WendelToBin1}
	\end{eqnarray}
	($\Leftarrow$) Assume $1<\beta<2$. We show $p_{n,K_a(n)} \rightarrow 1$ exponentially. Let $B_n \sim \mathrm{Bin}(K_a(n)-1,1/2)$. By (\ref{WendelToBin1}), we have $1-p_{n,K_a(n)}= \mathbb{P}[B_n \geq n]$. Since the threshold $n$ scales linearly with the number of trials $K_a(n)-1$ (considering $B_n$ as the summation of $K_a(n)-1$ i.i.d. $\mathrm{Bern}(1/2)$ random variables) and requires a macroscopic deviation from the mean, the probability $\mathbb{P}[B_n \geq n]$ is governed by large deviation principle \cite{Arratia1989, Dubhashi2009, Dembo1998}. By   Cram\'er's theorem on large deviation
\begin{eqnarray}
	&& \!\!\!\!\!\!\!\!\!\!\!\!\!\!\!\!\!\!\!\!\!  \mathbb{P}[B_n\ge n]= \nonumber\\
	&& \!\!\!\!\!\!\!\!\!\!\!\!\!\!\!\!\!\!\!\!\! \exp\!\left\{
-(K_a(n)-1)\,D\!\left(\frac{n}{K_a(n)-1}\bigg\|\frac12\right)
+o(K_a(n))
\right\}, \label{sigg33}
\end{eqnarray}
where $D(.\|.)$ is the relative entropy, resulting into
\begin{eqnarray}
	\lim_{n\to\infty}
-\frac{1}{n}\log\bigl(1-p_{n,K_a(n)}\bigr)
=
\beta\,D\!\left(\frac{1}{\beta}\bigg\|\frac12\right)
>0.
\end{eqnarray}
($\Rightarrow$) As discussed $p_{n,K_a(n)}=\mathbb{P}[B_n \leq n-1]$. Assume $p_{n,K_a(n)} \rightarrow 1$ exponentially as $n\rightarrow \infty$. We show that $1<\beta <2$ necessarily. If $\beta >2$, the threshold $(n-1)/(K_a(n)-1)$ converges to a constant less than $1/2$, so by Weak Law of Large Numbers (WLLN), the normalized sum $B_n/(K_a(n)-1)$ concentrates near $1/2$ and hence $p_{n,K_a(n)} = \mathbb{P}[B_n \leq n-1] \rightarrow 0$. If $\beta=2$, by Central Limit Theorem (CLT) and this fact that $(n-1-\mathbb{E}[B_n])/\sqrt{\mathrm{Var}(B_n)}$ tends to $0$, we get $p_{n,K_a(n)} \rightarrow 1/2$. Finally, when $0<\beta\leq 1$, the condition $K_a(n)-1 \leq n-1$ holds for all $n$, resulting into counting all binomial coefficients in (\ref{WendelToBin1}) and hence $p_{n,K_a(n)}=1$ for all $n$. The complete proof is given in Appendix A.
\end{proof}

We know that the sum rate across the channel is \( R_{\mathrm{sum}} = K_a(n)\log M_n / n = \beta d \log n \) where $d>2$. Under the assumed codebook geometry, the per user power is \( P \). So, the sum capacity of the Gaussian multiple-access channel is \( C_{\mathrm{sum}} = \frac{1}{2}\log(1+n\beta P) \sim \frac{1}{2}\log n \) as \( n \to \infty \). A necessary condition for reliable decoding is therefore \( R_{\mathrm{sum}} \leq C_{\mathrm{sum}} \), which implies \( \beta \leq \frac{1}{2d} \), and in particular \( \beta < \frac{1}{4} \) for sufficiently large \( n \). Moreover, by Theorem \ref{K_aHemispherical}, the selected \( K_a(n) \)-subset for this range of $\beta$ is hemispherical with probability one. Consequently, reliable decoding reduces to a search restricted to a hemisphere.



\section{Geometric Decoding}\label{GeoDec}
 
Under the many-access channel regime with $0<\beta<1/4$, we proved that the chosen $K_a(n)$-subset is a hemispherical set with high probability and reliable decoding is possible. Upon observing the channel output $\mbox{\boldmath $Y$}$, the first stage of decoding (pre-filtering) is to restrict the whole sphere based on the estimated direction $\hat{\mbox{\boldmath $u$}} = \mbox{\boldmath $Y$}/ \| \mbox{\boldmath $Y$}\|$. We define the \emph{limit hemisphere} with direction $\hat{\mbox{\boldmath $u$}}$ as
\begin{eqnarray}
	\hat{\mathcal{H}}\triangleq \left \{ \mbox{\boldmath $s$}_{j}  \in  \mathbb{S}^{n-1}: \langle \mbox{\boldmath $s$}_{j},  \hat{\mbox{\boldmath $u$}} \rangle \geq 0 \right \}.
\end{eqnarray}
While $\hat{\mathcal{H}}$ represents the ideal region, the noisy estimate $\hat{\mbox{\boldmath $u$}}$ is insufficient to precisely determine the actual hemisphere even as $n\rightarrow\infty$ (Theorem \ref{ThmConvToc}). Furthermore, restricting the search to $\hat{\mathcal{H}}$ may exclude relevant transmitted codewords (Theorem \ref{THMRetention}). To ensure high-probability retention of the codewords in the restricted space, we define a $\tau_n$-enlargement of the limit hemisphere consisting of a sequence of spherical caps $\{ \hat{\mathcal{H}}_{n}\}$ converging to $\hat{\mathcal{H}}$ as 
\begin{eqnarray}
	\hat{\mathcal{H}}_n \triangleq \left\{ \mbox{\boldmath$s$}_{j}\in \mathbb{S}^{n-1}: \langle \mbox{\boldmath $s$}_{j}, \hat{\mbox{\boldmath $u$}}\rangle \geq \tau_n \right\}, \label{SeqCaps}
\end{eqnarray}
for a sequence $\tau_n \rightarrow 0^{-}$. We characterize the convergence $\hat{\mathcal{H}}_{n} \rightarrow \hat{\mathcal{H}}$ using Hausdorff distance $d_H$ on the space of compact subsets of $\mathbb{S}^{n-1}$. Let $d(\mbox{\boldmath$x$}, \mbox{\boldmath $y$}) = \arccos\langle \mbox{\boldmath $x$}, \mbox{\boldmath $y$}\rangle$ be the geodesic distance on the sphere. Since $\hat{\mathcal{H}}\subset \hat{\mathcal{H}}_{n}$, the Hausdorff $d_H$ between these sets is defined as $d_{H}(\hat{\mathcal{H}}_{n}, \hat{\mathcal{H}}) \triangleq \sup_{\mbox{\scriptsize \boldmath $x$}\in \hat{\mathcal{H}}_{n}} \inf_{\mbox{\scriptsize \boldmath $y$}\in \hat{\mathcal{H}}} d(\mbox{\boldmath $x$}, \mbox{\boldmath $y$})$. The $\mathrm{sup}$ is achieved by $\mbox{\boldmath $x$}$ on the boundary of $\hat{\mathcal{H}}_{n}$ (where $\langle \mbox{\boldmath $x$}, \hat{\mbox{\boldmath $u$}}\rangle = \tau_n$). Also, the closest $\mbox{\boldmath $y$}$ on $\hat{\mathcal{H}}$ lies on the boundary where $\langle \mbox{\boldmath $y$}, \hat{\mbox{\boldmath $u$}}\rangle = 0$. Therefore, the geodesic distance simplifies to the angular distance between boundaries and $\hat{\mbox{\boldmath $u$}}$, resulting into $d_H(\hat{\mathcal{H}}_{n}, \hat{\mathcal{H}}) = \arccos \tau_n - \frac{\pi}{2}$, which converges to $0$ as $\tau_n \rightarrow 0^{-}$ and $n\rightarrow \infty$. Here, since we normalized the direction of spherical caps, we normalized the codewords containing in them likewise.

The second stage of decoding is applying the maximum likelihood to the codewords limited to $\{\hat{\mathcal{H}}_{n}\}$, formally for a fixed $n$

\begin{eqnarray}
	\hat{\mathcal{S}} = \arg \min_{\substack{\mathcal{S}\subset \{\hat{\mathcal{H}}_{n}\} \\ | \mathcal{S}| = K_a(n)}} \left \| \mbox{\boldmath $Y$} - \sum_{i\in \mathcal{S}}\mbox{\boldmath $x$}_{i} \right\| ^2. \label{ML}
\end{eqnarray}

\begin{thm}\label{ThmConvToc}
	Let $\mbox{\boldmath $u$}$ be the actual direction of hemisphere and $\| \mbox{\boldmath $u$}\| = 1$. Then
	\begin{eqnarray}
		\langle \mbox{\boldmath $u$}, \hat{\mbox{\boldmath $u$}} \rangle \xrightarrow{P} c  = \sqrt{\frac{2 \beta}{2 \beta + \pi}}\label{c}.
	\end{eqnarray}
\end{thm}
\begin{proof}
	The sketch of the proof is the following. As inspired by directional statistics \cite{DS1999}, we decompose $\mbox{\boldmath $x$}_{i} = t_i\mbox{\boldmath $u$} + \mbox{\boldmath $q$}_{i}$, where $t_i=\langle \mbox{\boldmath $x$}_{i},\mbox{\boldmath $u$}\rangle$, $\mbox{\boldmath $q$}_{i} \in \mathcal{U}^{\perp}$, and $\mathcal{U}^{\perp}$ is the set of vectors orthogonal to $\mbox{\boldmath $u$}$. Then, by Gaussian projection approximation lemma, it yields that $\sqrt{n}t$ converges in distribution to $\mathcal{N}(0,1)$ for $t=\langle \mbox{\boldmath $u$},\mbox{\boldmath $x$}/\sqrt{nP}\rangle$ with fixed $\mbox{\boldmath $u$}$ and uniformly distributed  $\mbox{\boldmath $x$}$ on sphere. Leveraging this, we discuss how each term in $\langle \mbox{\boldmath $Y$}/ \| \mbox{\boldmath $Y$} \|,\mbox{\boldmath $u$}\rangle$ converges in probability. The complete proof is given in Appendix B.
\end{proof}

We next define the \emph{retention} probability over $\{ \hat{\mathcal{H}}_{n}\}$ for a transmitted codeword as
\begin{eqnarray}
	&& p_{\mathrm{ret},n}(\tau_n)\label{retention}\\
	&& \ \ \ \ \triangleq \mathbb{P}[\langle \mbox{\boldmath $s$}_{i}, \hat{\mbox{\boldmath $u$}}\rangle \geq \tau_n | \text{$i$ is transmitted codeword index}]\nonumber,
\end{eqnarray}
which can be interpreted as the portion of sent codewords that will remain in the restricted search space.

\begin{thm}\label{THMRetention}
	Let $0<\beta<1/4$. For the defined retention probability in (\ref{retention}), we have $p_{\mathrm{ret},n}(\tau_n) \rightarrow 1,$
	if $\tau_n = -a_n/ \sqrt{n}$ where $a_n \rightarrow \infty$ and $a_n = o(\sqrt{n})$,
	and for the special case of searching only over $\hat{\mathcal{H}}$, we have
	\begin{eqnarray}\label{RetentionLimit}
		p_{\mathrm{ret},n}(0) \rightarrow \frac{1}{2} + \frac{1}{\pi} \arcsin c.
	\end{eqnarray}
\end{thm}
\begin{proof}
	The sketch of the proof is the following. Analogous to the decomposition used in Theorem \ref{ThmConvToc}, write
$\hat{\mbox{\boldmath $u$}}=c_n \mbox{\boldmath $u$}+\sqrt{1-c_n^2}\mbox{\boldmath $v$}_{n}$,
where $\mbox{\boldmath $v$}_{n}\in\mathcal{U}^\perp$ and $\|\mbox{\boldmath $v$}_{n}\|=1$. Here
$c_n=\langle \hat{\mbox{\boldmath $u$}},\mbox{\boldmath $u$}\rangle \xrightarrow{P} c$ by Theorem \ref{ThmConvToc}. Therefore,
\begin{equation}
	\left\langle \mbox{\boldmath $s$}_i,\hat{\mbox{\boldmath $u$}}\right\rangle
	=
	c_n \left\langle \mbox{\boldmath $s$}_i,\mbox{\boldmath $u$}\right\rangle
	+
	\sqrt{1-c_n^2}\left\langle \mbox{\boldmath $s$}_i,\mbox{\boldmath $v$}_n	\right\rangle .
\end{equation}

Define the score $T_n \triangleq \sqrt{n}\left\langle \mbox{\boldmath $s$}_i,\hat{\mbox{\boldmath $u$}}\right\rangle$, and threshold $\gamma_n \triangleq \sqrt{n}\,\tau_n$.
Then
\begin{eqnarray}
	p_{\mathrm{ret},n}(\tau_n)=\mathbb{P}[T_n\ge \gamma_n | \text{$i$ is the sent codeword index}].
\end{eqnarray}

We first consider the $\hat{\mathcal{H}}$ case. Note that for this case, there is no sequence of $\tau_n$, but the static threshold $0$. For fixed orthonormal vectors $(\mbox{\boldmath $u$},\mbox{\boldmath $v$}_n)$ and 
\begin{eqnarray}
	\mbox{\boldmath $s$}_i\sim \mathrm{Unif}	\Big\{\mbox{\boldmath $s$}\in\mathbb{S}^{n-1}:\langle \mbox{\boldmath $s$},\mbox{\boldmath $u$}\rangle\ge 0\Big\},
\end{eqnarray}
we have
\begin{equation}\label{JointLimitTau}
	\left(
	\sqrt{n}\left\langle \mbox{\boldmath $s$}_i,	\mbox{\boldmath $u$}\right\rangle,
	\sqrt{n}\left\langle \mbox{\boldmath $s$}_i,	\mbox{\boldmath $v$}_n\right\rangle
	\right)
	\xrightarrow{d}(|N_1|,N_2),
\end{equation}
where $N_1,N_2\sim\mathcal N(0,1)$ are independent. Indeed, after rotating coordinates so that
$\mbox{\boldmath $u$}=\mbox{\boldmath $e$}_{1}$ and $\mbox{\boldmath $v$}_n=\mbox{\boldmath $e$}_{2}$, the hemisphere representation
\begin{eqnarray}
	\mbox{\boldmath $x$}_i=\sqrt{nP}\,\frac{\mbox{\boldmath $P$}}{\|\mbox{\boldmath $P$}\|},
	\
	\mbox{\boldmath $P$}\sim \mathcal N(\mathbf 0,\mathbf I_n)\ \text{conditioned on } P_1\ge 0,
\end{eqnarray}
implies, by the WLLN and Slutsky’s lemma \cite{vander}, that \eqref{JointLimitTau} holds. Note that here $\mbox{\boldmath $P$}=(P_1,..., P_n)$.

Now apply the decomposition of $\hat{\mbox{\boldmath $u$}}$
\begin{eqnarray}
	T_n= c_n\sqrt{n}\left\langle \mbox{\boldmath $s$}_i,\mbox{\boldmath $u$}\right\rangle
	+
	\sqrt{1-c_n^2}\sqrt{n}\left\langle \mbox{\boldmath $s$}_i,\mbox{\boldmath $v$}_n\right\rangle. \label{18}
\end{eqnarray}
Combining (\ref{18}) with \eqref{JointLimitTau} , by Slutsky’s lemma, it yields
\begin{eqnarray}
	T_n \xrightarrow{d} W_0\triangleq c|N_1|+\sqrt{1-c^2}\,N_2.
\end{eqnarray}

Since $W_0$ is a continuous random variable, the probability $\mathbb P[W_0=0]=0$.
Hence, by the Portmanteau lemma \cite{vander},
\begin{eqnarray}
	&& p_{\mathrm{ret},n}(\tau_n)=
\mathbb P[T_n\ge 0 | \text{$i$ is sent}]
\rightarrow \mathbb P[W_0\ge 0 ].
\end{eqnarray}
With simple calculations, we get
\begin{eqnarray}
	p_{\mathrm{ret},n}(0)\to \mathbb P[W_0\ge 0]=
\frac12+\frac{1}{\pi}\arcsin c.
\end{eqnarray}

Finally, if $\tau_n=-a_n/\sqrt n$ with $a_n\to\infty$ and $a_n=o(\sqrt n)$, then $\tau_n\to 0^-$ while $\gamma_n=\sqrt n\,\tau_n=-a_n\to -\infty$. In this case, there is no truncation ($P_1\geq 0$). Therefore, the pair in (\ref{JointLimitTau}) converges in distribution to $(N_1,N_2)$, resulting into $T_n \xrightarrow{d} W \triangleq cN_1+\sqrt{1-c^2}N_2$. Leveraging this, by Prohorov's theorem \cite{vander}, we know that $T_n=O_P(1)$, which follows $\mathbb P[T_n\ge \gamma_n]\to 1$ as $\gamma_n \rightarrow -\infty$ . The more detailed proof is given in Appendix C.
\end{proof}

Before applying maximum-likelihood (ML) estimator to the filtered codewords, an immediate necessary condition is that at least $K_a(n)$ codewords lie within $\{ \hat{\mathcal{H}}_{n}\}$. This raises the question: does the probability of having fewer than $K_a(n)$ codewords in the sequence of spherical caps $\mathbb{P}\left[|\hat{\mathcal{H}}_{n}| < K_a(n)\right]$ tend to zero as $n\rightarrow \infty$? In the following theorem, we prove that under the many-access assumption with $0<\beta<2$, this probability indeed converges to $0$ as $n\rightarrow \infty$.

\begin{thm}\label{H_nK_a}
	Consider the described setup, then $\mathbb{P}\left[|\hat{\mathcal{H}}_n| <  K_a(n) \right] \rightarrow 0$ as $n \rightarrow \infty$.
\end{thm}
\begin{proof}
	The sketch of the proof is the following. Fix $n$. We know that $\hat{\mathcal{H}}_{n}$ is larger than $\hat{\mathcal{H}}$, and hence
	\begin{eqnarray}
		\mathbb{P}\left[| \hat{\mathcal{H}}_{n}| < K_a(n)\right] < \mathbb{P}\left[| \hat{\mathcal{H}}| < K_a(n)\right]. 
    \end{eqnarray}	 
     Our goal in this proof is then shifted to $\mathbb{P}\left[| \hat{\mathcal{H}}| < K_a(n)\right] \rightarrow 0$. W.L.O.G, we assume $\mathcal{S}=[K_a(n)]$. We split the cardinality $| \hat{\mathcal{H}}|$ into 
	\begin{eqnarray}
		| \hat{\mathcal{H}} | = \underbrace{\!\! \sum_{i=1}^{K_a(n)} \mathbf{1}\{\langle \mbox{\boldmath $s$}_{i}, \hat{\mbox{\boldmath $u$}}\rangle \geq 0 \}}_{\triangleq H_{\mathrm{true}}^{(n)}} +\!\! \underbrace{\sum_{j=K_a(n)+1}^{M_n}\!\!\!\!\! \mathbf{1}\{\langle \mbox{\boldmath $s$}_{j},\hat{\mbox{\boldmath $u$}}\rangle \geq 0 \}}_{\triangleq H_{\mathrm{other}}^{(n)}}.\label{CarH}
	\end{eqnarray}
	A simple observation is that conditional on $\mbox{\boldmath $Y$}$ (or 
$\hat{\mbox{\boldmath $u$}}$), the remaining $M_n-K_a(n)$ points
$\mbox{\boldmath $s$}_{K_a(n)+1},\ldots,\mbox{\boldmath $s$}_{M_n}$ are i.i.d. uniform on
$\mathbb S^{n-1}$ and independent of $\hat{\mbox{\boldmath $u$}}$ (or $\mbox{\boldmath $Y$}$). Therefore, conditional on
$\mbox{\boldmath $Y$}$, the count $H^{(n)}_{\mathrm{other}}$ has distribution $H^{(n)}_{\mathrm{other}} \mid \mbox{\boldmath $Y$} \sim \mathrm{Bin}\!\left(M_n-K_a(n),\,p_n\right)$, where $p_n = \mathbb P\!\left[\sqrt{n}\left\langle \mbox{\boldmath $s$}_{j},\hat{\mbox{\boldmath $u$}}\right\rangle \ge 0\right]$.

By rotational symmetry, we can rotate the direction of $\hat{\mbox{\boldmath $u$}}$ to $\mbox{\boldmath $e$}_{1}$ without changing the distribution of $\sqrt{n}\langle \mbox{\boldmath $s$}_{j}, \hat{\mbox{\boldmath $u$}}\rangle$ \cite{DS1999}, which results into $\sqrt{n}\langle \mbox{\boldmath $s$}_{j}, \hat{\mbox{\boldmath $u$}}\rangle \xrightarrow{d}\mathcal{N}(0,1)$. Hence, by Portmanueau lemma for all $j\in \{K_a(n)+1,...,M_n\}$, it yields
\begin{equation}
	p_n= \mathbb P\!\left[\sqrt{n}\langle\mbox{\boldmath $s$}_{j},	\hat{\mbox{\boldmath $u$}}\rangle \ge 0\right]
\longrightarrow \mathbb P[\mathcal{N}(0,1) \geq 0]=\frac{1}{2}.
\label{pnlimit1}
\end{equation}
A more rigorous way to show (\ref{pnlimit1}) is based on Lindeberg-Feller CLT \cite{Bilingsly2012} which is given in Appendix D.

Taking expectation from both sides of \eqref{CarH}, we have
\begin{eqnarray}
	\mathbb E\!\left[\left|\hat{\mathcal H}\right|\right]
	&& \!\!\!\!\!\!\!\!\!\!\!\! = \mathbb E\!	\left[\sum_{i=1}^{K_a(n)} \mathbf{1}\!\left\{\left\langle \mbox{\boldmath $s$}_i,\hat{\mbox{\boldmath $u$}}\right\rangle \ge 0\right\}\right]
	+ \mathbb E\!\left[\mathbb E\!\left[H^{(n)}_	{\mathrm{other}} \middle| \mbox{\boldmath $Y$}\right]\right] \nonumber\\
	&& \!\!\!\!\!\!\!\!\!\!\!\!\!\!\!\!\!\!\!\!\!\!\!\!\!\!	\!\!\!\!\!\!\!\!\!\! = K_a(n)\,\mathbb P\!\left[\left\langle \mbox{\boldmath $s$}_i,	\hat{\mbox{\boldmath $u$}}\right\rangle \ge 0 \,\middle|\, \mbox{\boldmath $s$}_i \ \text{is sent} \right]
\! + \! \big(M_n-K_a(n)\big)p_n. \label{EXPH_cap}
\end{eqnarray}
As proved in Theorem \ref{THMRetention}, we know that
\begin{equation}
p_{\mathrm{ret},n}(0)=\mathbb P\!\left[\left\langle \mbox{\boldmath $s$}_i,\hat{\mbox{\boldmath $u$}}\right\rangle \ge 0 \,\middle|\, \mbox{\boldmath $s$}_i \ \text{is sent} \right]
\rightarrow \frac{1}{2}+\frac{1}{\pi}\arcsin c,
\end{equation}
as $n \rightarrow \infty$. 
Hence, for every $\epsilon>0$, there exists $n_0>0$ such that for all $n\ge n_0$,
\begin{equation}
	\frac{1}{2}+\frac{1}{\pi}\arcsin c - \epsilon < p_{\mathrm{ret},n}(0) <\frac{1}{2}+\frac{1}{\pi}\arcsin c + \epsilon.
\label{DefLimit_cap}
\end{equation}
By \eqref{EXPH_cap} and \eqref{DefLimit_cap}, along with $M_n= n^d, d>2$ and $K_a(n)=\beta n$, for all large $n$, we obtain $\mathbb E\!\left[\left|\hat{\mathcal H}\right|\right]\geq (\beta n) \left(\frac{1}{2}+\frac{1}{\pi}\arcsin c -\epsilon \right) + (n^d -\beta n)\left(\frac{1}{2}-\epsilon\right) $.
Since the dominant term is $n^d (1/2-\epsilon)$ and $K_a(n)=\beta n$, then $\mathbb E\!\left[\left|\hat{\mathcal H}_n\right|\right] - K_a(n) > c' n^d$ for some constant $c'>0$ and all sufficiently large $n$. We next upper bound the desired probability as
	\begin{eqnarray}
		&&\!\!\!\!\!\!\!\! \mathbb{P}\left[|\hat{\mathcal{H}}_{n}| < K_a(n) \right] \leq \mathbb{P}\left[| \hat{\mathcal{H}}_{n}| -  \mathbb{E}\left[|\hat{\mathcal{H}}_{n}|\right] < -c^{\prime}n^d \right]\nonumber\\
		&& \!\!\!\!\!\!\!\!\!\! \leq \mathbb{P}\left[H_{\mathrm{true}}^{(n)} - \mathbb{E}\left[H_{\mathrm{true}}^{(n)}\right] < -\frac{c^{\prime}}{2}n^d \right] \label{bound_based_on_split0}\\
		&& \ \ \ \ \ \ \ \ \ \ \ + \mathbb{P}\left[H_{\mathrm{other}}^{(n)} - \mathbb{E}\left[H_{\mathrm{other}}^{(n)}\right] < -\frac{c^{\prime}}{2}n^d \right].\label{bound_based_on_split00}
	\end{eqnarray}
	By McDiarmid's inequality, we prove that (\ref{bound_based_on_split0}) is upper bounded by $e^{-\tilde{c}_1n^{2d-1}}$ for some constant $\tilde{c}_{1}>0$. By Hoeffding's inequality, we show  (\ref{bound_based_on_split00}) is also upper bounded by $e^{-\tilde{c}_2n^d}$ for some constant $\tilde{c}_{2}>0$. The more detailed proof is given in Appendix D.
\end{proof}

\section{Per-User Error Probabilities}

The pre-filtering stage begins upon observing the channel output $\mbox{\boldmath $Y$}$. We first estimate the limit hemisphere direction as $\hat{\mbox{\boldmath $u$}}=\mbox{\boldmath $Y$}/  \| \mbox{\boldmath $Y$}\|$ and restrict the search for the $K_a(n)$-subset with $0<\beta<1/4$ to the sequence of spherical caps $\{ \hat{\mathcal{H}}_{n}\}$ converges to $\hat{\mathcal{H}}$ in Hausdorff metric.
Here, we define pre-filtering per-user probability error ($\mathrm{PUPE}_{p}(n)$) for a fixed $n$ as 
\begin{eqnarray}
	\mathrm{PUPE}_{p}(n,\tau_n) \triangleq \frac{1}{K_a(n)} \sum_{i=1}^{K_a(n)} \mathbb{P}\left[\langle \mbox{\boldmath $s$}_{i}, \hat{\mbox{\boldmath $u$}}\rangle < \tau_n \right].\label{DefPUEP_p}
\end{eqnarray}

\begin{thm} \label{THMprePUEP}
	For the described model with $\tau_n=-a_n/\sqrt{n}$ where $a_n\rightarrow \infty$ and $a_n=o(\sqrt{n})$, the per-user probability error $\mathrm{PUEP}_{p}(n,\tau_n) \rightarrow 0$ as $n\rightarrow \infty$. If we restrict our search to only $\hat{\mathcal{H}}$, then $\mathrm{PUPE}_{p}(n,0)  \rightarrow \frac{1}{2} - \frac{1}{\pi}\arcsin c$.
\end{thm}

\begin{proof}
	According to the codebook symmetry and as a direct consequence of Theorem \ref{THMRetention}, the results are immediate.
\end{proof}

The per-user probability error regarding ML step ($\mathrm{PUPE}_{ML}(n)$) is the expected function of incorrectly identified codewrods defined as
\begin{eqnarray}\label{DefPUEPML}
	\mathrm{PUPE}_{ML}(n) \triangleq \frac{1}{K_a(n)} \sum_{\ell = 1}^{K_a(n)} \ell \ \mathbb{P}\left[| \mathcal{S}\backslash \hat{\mathcal{S}}| = \ell \right].
\end{eqnarray}
According to Theorems \ref{THMRetention} and \ref{THMprePUEP}, since it is impossible to recover the transmitted set by searching only over $\hat{\mathcal{H}}$, we dedicate our analysis on $\mathrm{PUPE}_{ML}(n)$ only to $\{ \hat{\mathcal{H}}_{n}\}$. Moreover, with number of active users $K_a(n)=\beta n$ and codebook size $M_n = n^d , d>2$, the collision probability $p_0(n)$ (different active users choose the same message) 
\begin{eqnarray}
	p_0(n) \triangleq \mathbb{P}[\text{message collision}] \leq \frac{1}{M_n}{K_a(n) \choose 2} \!\leq \!\frac{\beta ^2 n^2}{2n^d} \rightarrow 0. \nonumber
\end{eqnarray}

\begin{thm}
	Consider the described setup with search space $\{ \hat{\mathcal{H}}_{n}\}$. Then, the error exponent for $\mathrm{PUPE}_{ML}(n)$ is
		\begin{eqnarray}
	 		\lim_{n\rightarrow \infty}-\frac{1}{n}\log 		\mathrm{PUPE}_{ML}(n) = \frac{P}{4}.
		\end{eqnarray}
\end{thm}

\begin{proof}
	Fix $n$. The error occurs if decoder during ML step chooses $\mathcal{S}^{\prime}$ with $| \mathcal{S}^{\prime}| = n\beta$ over the true set $\mathcal{S}$. Hence,
	\begin{eqnarray}
		\left\| \mbox{\boldmath $Y$}-\sum_{i\in \mathcal{S}^{\prime}}\mbox{\boldmath $x$}_{i} \right\| ^2 \leq \left\| \mbox{\boldmath $Y$} - \sum_{i\in \mathcal{S}}\mbox{\boldmath $x$}_{i}\right\| ^2. \label{S_primeS}
	\end{eqnarray}
	Let $| \mathcal{S}\backslash \mathcal{S}^{\prime}| = \ell$ and define 
	\begin{eqnarray}
		\Delta_{\ell,n} \triangleq \sum_{i\in \mathcal{S}}\mbox{\boldmath $x$}_{i} - \sum_{j\in \mathcal{S}^{\prime}}\mbox{\boldmath $x$}_{j}. \label{Delta}
	\end{eqnarray}
	Then we can simplify (\ref{S_primeS}) to 
	\begin{eqnarray}
		\langle \mbox{\boldmath $Z$}, \Delta_{\ell ,n}\rangle \leq -\frac{\| \Delta_{\ell ,n}\| ^2}{2}.
	\end{eqnarray}
	Since $\mbox{\boldmath $Z$} \sim \mathcal{N}(\mbox{\boldmath $0$}, \mathbf{I}_{n})$, then $\langle \mbox{\boldmath $Z$}, \Delta_{\ell ,n}\rangle \sim \mathcal{N}(0, \| \Delta_{\ell,n}\|^2)$. We can simply write the probability of pairwise error, conditioned on $\Delta_{\ell,n}$ as 
	\begin{eqnarray}
		\mathbb{P}\left[\mathcal{S} \rightarrow \mathcal{S}^{\prime} \middle \vert \Delta_{\ell ,n}\right] = Q\left(\frac{\| \Delta_{\ell,n}\|}{2}\right) \stackrel{(a)}{\leq} e^{-\frac{\| \Delta_{\ell,n}\|^2}{8}}, \label{UppQ}
	\end{eqnarray}
	where $(a)$ is by Chernoff bound on $Q$-function.
	
	For an inactive user $\mbox{\boldmath $x$}_{j} \notin 	\mathcal{S}$, we already discussed in the proof of Theorem \ref{H_nK_a} that $\langle \mbox{\boldmath $x$}_{j}, \hat{\mbox{\boldmath $u$}}\rangle \xrightarrow{d} \mathcal{N}(0,P)$. So, the probability that an unsent codeword falls inside the search space is $p_n$ with limit  $1/2$ as in (\ref{pnlimit1}). Now, we can model the presence of an unsent codewrod in $\{ \hat{\mathcal{H}}_{n}\}$ by a Bernoulli random variable with probability $1/2$ as $n\rightarrow \infty$. Hence, by WLLN, the number of unsent codewords in $\{ \hat{\mathcal{H}}_{n}\}$ concentrates tightly around $(M_n-K_a(n))/2 = (n^d - \beta n) /2$. 
	
	To find the total probability of an $\ell$-fold error $\mathbb{P}\left[ |\mathcal{S}\backslash \mathcal{S}^{\prime}| = \ell \right]$, we account for every possible way that decoder can construct $\mathcal{S}^{\prime}$ differing $\mathcal{S}$ by exactly $\ell$ codewords. It requires choosing $\ell$ sent codewords to drop from $\mathcal{S}$ and $\ell$ unsent codewords to pick from the unsent codewords in $\{ \hat{\mathcal{H}}_{n}\}$. Hence, by union bound and (\ref{UppQ}), it yields
	\begin{eqnarray}
		\mathbb{P}\left[| \mathcal{S}\backslash \hat{\mathcal{S}}| = \ell \right] \leq {\beta n \choose \ell}{(n^d-\beta n)/2 \choose \ell} \mathbb{E}\left[e^{-\frac{\| \Delta_{\ell,n}\|^2}{8}}\right]. \label{ell}
	\end{eqnarray}
	The vector $\Delta_{\ell,n}$ is the sum of $\ell$ sent codewords and the negative of $\ell$ unsent codewords. By expanding the squared norm, we obtain
	\begin{eqnarray}
		\| \Delta_{\ell,n}\|^2 = 2\ell nP + \sum_{i\neq j}\langle \mbox{\boldmath $x$}_{i}, \mbox{\boldmath $x$}_{j}\rangle. \label{ExpansionDelta}
	\end{eqnarray}
	 As $n\rightarrow \infty$, any two independent vectors on $\mathbb{S}^{n-1}$ are asymptotically orthogonal \cite{Gorban2018}, \cite{Cai2013}. Formally, for our problem where the sphere radius is $\sqrt{nP}$, we have $\frac{\langle \mbox{\boldmath $x$}_{i}, \mbox{\boldmath $x$}_{j}\rangle}{nP} \xrightarrow{P} 0$. 
	 By (\ref{ExpansionDelta}), we get
	 \begin{eqnarray}
	 	\frac{\| \Delta_{\ell,n}\|^2}{n} \xrightarrow{P} 2\ell P,
	 \end{eqnarray}
	 which results into $\mathbb{P}\left[\left | \| \Delta_{\ell,n}\|^2 - 2\ell nP\right | \geq \epsilon n\right] \rightarrow 0$ for all $\epsilon >0$ according to the definition of convergence in probability. Defining $\mathcal{D}_n \triangleq \{ \left | \| \Delta_{\ell,n}\|^2 - 2\ell nP\right | \leq \epsilon n \}$, we split
	 \begin{eqnarray}
	 	\mathbb{E}\left[e^{-\frac{\| \Delta_{\ell,n}\|^2}{8}}\right]\! = \! \mathbb{E}\left[e^{-\frac{\| \Delta_{\ell,n}\|^2}{8}} \mathbf{1}_{\mathcal{D}_{n}}\right] \! + \! \mathbb{E}\left[e^{-\frac{\| \Delta_{\ell,n}\|^2}{8}} \mathbf{1}_{\mathcal{D}_{n}^{c}}\right].
	 \end{eqnarray}
	 On $\mathcal{D}_{n}$, we have
	 \begin{eqnarray}
	 	e^{-\frac{1}{8}(2\ell P +\epsilon)n} \leq e^{\frac{-\| \Delta_{\ell ,n}\| ^2}{8}} \leq e^{-\frac{1}{8}(2\ell P - \epsilon)n}.
	 \end{eqnarray}
	 Thus,
	\begin{eqnarray}
		e^{-\frac{1}{8}(2\ell P +\epsilon)n} \mathbb{P}[\mathcal{D}_{n}] \leq \mathbb{E}\left[e^{-\frac{\| \Delta_{\ell,n}\|^2}{8}} \mathbf{1}_{\mathcal{D}_{n}}\right] \leq e^{-\frac{1}{8}(2\ell P - \epsilon)n}.
	\end{eqnarray}
	Since $\mathbb{P}[\mathcal{D}_{n}] \rightarrow 1$, then 
	\begin{eqnarray}
		\mathbb{E}\left[e^{-\frac{\| \Delta_{\ell,n}\|^2}{8}} \mathbf{1}_{\mathcal{D}_{n}}\right] = e^{-\frac{1}{4}\ell Pn + o(n)}. \label{ELeft}
	\end{eqnarray}
	Because $\| \Delta_{\ell,n}\|^2 \geq 0$, then $0\leq e^{-\frac{1}{8}\| \Delta_{\ell,n}\|^2} \leq 1$. Therefore,
	\begin{eqnarray}
		0\leq \mathbb{E}\left[e^{-\frac{\| \Delta_{\ell,n}\|^2}{8}} \mathbf{1}_{\mathcal{D}_{n}^{c}}\right] \leq \mathbb{P}[\mathcal{D}_{n}^{c}] \rightarrow 0.\label{ERight}
	\end{eqnarray}
	Combining the results in (\ref{ELeft}) and (\ref{ERight}), for sufficiently large $n$, we have
	\begin{eqnarray}
		\mathbb{E}\left[e^{-\frac{\| \Delta_{\ell,n}\|^2}{8}}\right] = e^{-\frac{1}{4}\ell Pn + o(n)}.\label{FinE}
	\end{eqnarray}
	Now, substituting (\ref{FinE}) into (\ref{ell}) and according to (\ref{DefPUEPML}), we get
	\begin{eqnarray}
		&& \mathrm{PUPE}_{ML}(n) \leq \\
		&&\frac{1}{n\beta} \sum_{\ell = 1}^{n\beta} \ell {\beta n \choose \ell}{(n^d-\beta n)/2 \choose \ell} e^{-\frac{1}{4}\ell nP + o(n)} \nonumber\\
		&& \stackrel{(c)}{=}\sum_{\ell =1}^{\beta n} \underbrace{{\beta n-1 \choose \ell -1}{(n^d-\beta n)/2	\choose \ell}	e^{-\frac{1}{4}\ell nP +o(n)}}_{\triangleq E_\ell}	\nonumber,
	\end{eqnarray}
	where $(c)$ follows from $\ell{\beta n \choose \ell} = n \beta {\beta n-1 \choose \ell-1}$. It is clear that
	\begin{eqnarray}
		\max_{\ell}E_\ell \leq \mathrm{PUPE}_{ML}(n) \leq \beta n \max_{\ell}E_{\ell},
	\end{eqnarray}
	resulting into
	\begin{eqnarray}
		\frac{-1}{n}(\log \beta n \max_{\ell}E_\ell)\! \leq\!  \frac{-1}{n}\log \mathrm{PUPE}_{ML}(n)\! \leq \! \frac{-1}{n}\log \max_{\ell}E_{\ell}\nonumber.
	\end{eqnarray}
	Since $\log(\beta n) / n \rightarrow 0$, both sides converge to the same value, and hence
	\begin{eqnarray}
		\lim_{n\rightarrow \infty}\frac{-1}{n}\log \mathrm{PUPE}_{ML}(n)=\lim_{n\rightarrow \infty}\frac{-1}{n}\log \max_{\ell}E_{\ell}. \label{48}
	\end{eqnarray}
	We examine the ratio of successive terms as the following
	\begin{eqnarray}
		&&\!\!\!\! \frac{E_{\ell +1}}{E_{\ell}} = \frac{\beta n -\ell}{\ell (\ell +1)}\left(\frac{n^d-\beta n}{2}-\ell \right)e^{-\frac{nP}{4}}\\
		&&\ \ \ \ \ \ \approx \left(\frac{\beta n -\ell}{2\ell (\ell +1)}\right)n^d e^{-\frac{nP}{4}}.\label{FUC}
	\end{eqnarray}
	The ratio in (\ref{FUC}) goes to $0$ as $n\rightarrow \infty$, and hence for all sufficiently large $n$, the ratio ${E_{\ell +1}}/{E_{\ell}} <1$, which results into $\max_{\ell}E_{\ell} = E_1$. Finally 
\begin{eqnarray}
	\frac{-1}{n}\log E_1 = \frac{P}{4}-\frac{o(n)}{n}-\frac{1}{n}\log \left(\frac{n^d - \beta n}{2}\right)\rightarrow \frac{P}{4}, \label{51}
\end{eqnarray}
as $n \rightarrow \infty$. Comparing (\ref{48}) and (\ref{51}) completes the proof.
\end{proof}

\section{Conclusion}

We analyzed subset recovery decoding in the Gaussian many-access channel with a spherical codebook of size $n^d, d>2$ whose codewords are drawn uniformly at random from $\mathbb{S}^{n-1}(\sqrt{nP})$. We prove that, for $0<\beta<2$, the randomly drawn $K_a(n)$-subset of points on sphere is hemispherical with high probability for large $n$. We further show that reliable decoding limited the range of $\beta$ to $0<\beta<1/4$, which is in the hemispherical concentration range. Leveraging this geometry, for reliable decoding, it suffices to search over the sequence of spherical caps $\{ \hat{\mathcal{H}}_{n}\}$ converging to hemisphere $\hat{\mathcal{H}}$ in Hausdorff distance, rather than over the entire sphere. By applying the ML estimator to this reduced candidate set, the decoder can recover the transmitted subset by per-user probability error goes exponentially to zero with decay rate $P/4$. As shown, Wendel's theorem gives a range of $\beta$ that is too coarse. A more precise relation between $\beta$, especially for $\beta<1/4$, and the spherical-cap angle of the reduced search space is a direction for future research.

\appendices

\section{Proof of Theorem \ref{K_aHemispherical}} \label{ProofK_aHemispherical}

\begin{proof}
We prove the two directions separately.

\medskip
\noindent\textbf{($\Leftarrow$)} Assume that $1<\beta<2$. We show $p_{n,K_a(n)}$ converges to $1$ exponentially.
By Wendel's theorem,
\begin{equation}\label{eq:wendel}
	p_{n,K_a(n)}= \frac{1}{2^{K_a(n)-1}}
\sum_{i=0}^{n-1}\binom{K_a(n)-1}{i}.
\end{equation}
Equivalently, if
\begin{eqnarray}
	B_n \sim \mathrm{Bin}\!\left(K_a(n)-1,\frac12\right),
\end{eqnarray}
then
\begin{eqnarray}
	p_{n,K_a(n)}=\mathbb{P}[B_n\le n-1],
\end{eqnarray}
resulting into
\begin{eqnarray}
	1-p_{n,K_a(n)}=\mathbb{P}[B_n\ge n],
\end{eqnarray}
which constitutes an upper-tail binomial event strictly above the expected value $\mathbb{E}[B_n]= (K_a(n)-1)/2$, because $1<\beta<2$ and the normalized threshold satisfies
\begin{eqnarray}
	\frac{n}{K_a(n)-1}\to \frac{1}{\beta}\in\left(\frac12,1\right), \qquad \mathrm{as} \ n \rightarrow \infty. \label{sig3}
\end{eqnarray}
We know that $B_n$ is the summation of $K_a(n)-1$ i.i.d. $\mathrm{Bern}(1/2)$ random variables.
Since threshold $n$ scales linearly with the number of trials $K_a(n)-1$ and requires a macroscopic deviation from the mean, the probability of this event is governed by the large deviation principle \cite{Arratia1989, Dubhashi2009, Dembo1998}. More precisely, by Cram\'er's theorem on large deviation
\begin{eqnarray}
	 && \lim_{K_a(n) \rightarrow \infty} \frac{1}{K_a(n)-1} \log \mathbb{P}\left[\frac{B_n}{K_a(n)-1} \! \geq \! \frac{n}{K_a(n)-1}\right]\! \nonumber \\
	  && \ \ \ \ = -\inf_{x\geq \frac{1}{\beta}} I(x),\nonumber
\end{eqnarray}
where for $\mathrm{Bern}(1/2)$, 
\begin{eqnarray}
	I(x) = D(x\|1/2)=x\log(2x)+(1-x)\log \bigl(2(1-x)\bigr)
\end{eqnarray}
is the binary relative entropy. Since $1/2 < 1/\beta <1$, this $I(x)$ is increasing for $x>1/2$, so
\begin{eqnarray}
	\inf_{x \geq \frac{1}{\beta}} I(x) = I\left(\frac{1}{\beta}\right).
\end{eqnarray}
Hence,
\begin{eqnarray}
	&& \!\!\!\!\!\!\!\!\!\!\!\!\!\!\!\!\!\!\!\!\!  \mathbb{P}[B_n\ge n]= \nonumber\\
	&& \!\!\!\!\!\!\!\!\!\!\!\!\!\!\!\!\!\!\!\!\! \exp\!\left\{
-(K_a(n)-1)\,D\!\left(\frac{n}{K_a(n)-1}\bigg\|\frac12\right)
+o(K_a(n))
\right\}. \label{sigg3}
\end{eqnarray}
Note that Cram\'er's theorem applies to fixed thresholds, but because the rate function is continuous, the same asymptotic holds for $n/(K_a(n)-1) \rightarrow 1/\beta$.

Finally, by (\ref{sigg3})
and this fact that $D(x\|1/2)$ is continuous and strictly positive for $x\neq 1/2$, we obtain
\begin{eqnarray}
	\lim_{n\to\infty}
-\frac{1}{n}\log\bigl(1-p_{n,K_a(n)}\bigr)
=
\beta\,D\!\left(\frac{1}{\beta}\bigg\|\frac12\right)
>0.
\end{eqnarray}
Hence there exists a constant $c>0$ such that
\begin{eqnarray}
	1-p_{n,K_a(n)}\le e^{-cn}
\end{eqnarray}
for all sufficiently large $n$. Therefore $p_{n,K_a(n)}\to 1$ exponentially.

\medskip
\noindent\textbf{($\Rightarrow$)} Assume that $p_{n,K_a(n)}$ converges to $1$ exponentially as $n\to\infty$. We show that necessarily $1<\beta<2$.

We first exclude $\beta\ge 2$. If $\beta>2$, then
\begin{eqnarray}
	\frac{n-1}{K_a(n)-1}\to \frac{1}{\beta}<\frac12, \qquad \mathrm{as} \ n \rightarrow \infty. \label{betaBigger2}
\end{eqnarray}
We know that $B_n$ is the summation of $K_a(n)-1$ i.i.d. Bernoulli random variables with probability $1/2$. Hence, by Weak Law of Large Numbers (WLLN), we have
\begin{eqnarray}
	\frac{B_n}{K_a(n)-1}\xrightarrow{P} \frac{1}{2},
\end{eqnarray}
which for every $\epsilon >0$, results into
\begin{eqnarray}
	\mathbb{P}\left[\left \vert \frac{B_n}{K_a(n)-1} -\frac{1}{2}\right \vert > \epsilon\right] \rightarrow 0,
\end{eqnarray}  
as $n \rightarrow \infty$. Now by the basic definition of limit for (\ref{betaBigger2}), it yields that there exists $\epsilon^\prime>0$ such that
\begin{eqnarray}
	\frac{1}{\beta}-\epsilon^\prime <\frac{n-1}{K_a(n)-1} < \frac{1}{\beta}+\epsilon^\prime < \frac{1}{2}+\epsilon^{\prime} , \label{epprime}
\end{eqnarray}
for large enough $n$. Setting $\epsilon = \epsilon^\prime$ and by (\ref{epprime}), we have the following chain
\begin{eqnarray}
	&&\!\!\!\!\!\!\!\!\! p_{n,K_a(n)}=\mathbb{P}[B_n \leq n-1] = \mathbb{P}\left[\frac{B_n}{K_a(n)-1} \leq \frac{n-1}{K_a(n)-1}\right]\nonumber\\
	&& \ \ \ \ \ \ \ \ \ \ \ \ \ \ \ \ \ \ \ \ \ \ \ \  < \mathbb{P}\left[\frac{B_n}{K_a(n)-1} < \frac{1}{2}+\epsilon\right] \nonumber \\
	&& \ \ \ \ \ \ \ \ \ \ \ \ \ \ \ \ \ \ \ \ \ \ \ \   < \mathbb{P}\left[\left \vert \frac{B_n}{K_a(n)-1}-\frac{1}{2}\right\vert >\epsilon \right] \rightarrow 0. \nonumber
\end{eqnarray}
Hence $p_{n,K_a(n)}$ cannot converge to $1$.

If $\beta=2$, then
\begin{eqnarray}
	\frac{n}{K_a(n)-1}\to \frac12.
\end{eqnarray}
Let $\mu_n \triangleq \mathbb{E}[B_n] = (K_a(n)-1)/2$ and $\sigma_n^2 \triangleq \mathrm{Var}(B_n)= (K_a(n)-1)/4$. By the central limit theorem (CLT), we get
\begin{eqnarray}
	\frac{B_n-\mu_n}{\sigma_n} \xrightarrow{d} \mathcal{N}(0,1).
\end{eqnarray}
Rewriting the probability
\begin{eqnarray}
	\mathbb{P}[B_n \leq n-1]=\mathbb{P}\left[\frac{B_n-\mu_n}{\sigma_n} \leq \frac{n-1-\mu_n}{\sigma_n}\right]. \label{sig1}
\end{eqnarray}
Since $\beta =2$ and $K_a(n)=\beta n$, then 
\begin{eqnarray}
	\frac{n-1-\mu_n}{\sigma_n} \rightarrow 0, \qquad \mathrm{as}\ n\rightarrow \infty. \label{sig2}
\end{eqnarray}
According to (\ref{sig1}) and (\ref{sig2}), by continuity of standard normal CDF at $0$, we have
\begin{eqnarray}
	p_{n,K_a(n)}=\mathbb{P}[B_n\leq n-1] \rightarrow \frac{1}{2}.
\end{eqnarray}
So, again $p_{n,K_a(n)}$ does not converge to $1$.

Therefore, exponential convergence to $1$ is impossible when $\beta\ge 2$.

Now consider $0<\beta\le 1$. Then for all sufficiently large $n$, we have
\begin{eqnarray}
	K_a(n)-1\le n-1,
\end{eqnarray}

and the sum in \eqref{eq:wendel} contains all binomial coefficients. Consequently, $p_{n,K_a(n)}=1$ for all $n$.
\end{proof}

\section{Proof of Theorem \ref{ThmConvToc}} \label{ProofThmconvToc}

\begin{proof}
	Let $\mbox{\boldmath $s$}_{i} = \mbox{\boldmath $x$}_{i}/\sqrt{nP}$, where $\mbox{\boldmath $x$}_{i}$ is a codeword (point) uniformly distributed on $\mathbb{S}^{n-1}(\sqrt{nP})$, and $t_i = \langle \mbox{\boldmath $s$}_{i}, \mbox{\boldmath $u$}  \rangle$. We denote the set of vectors that are orthogonal to $\mbox{\boldmath $u$}$ by $\mathcal{U}^{\perp}$. We then decompose each point $\mbox{\boldmath $s$}_{i}$ as
	\begin{eqnarray}
		\mbox{\boldmath $s$}_{i} = t_i \mbox{\boldmath $u$} + \mbox{\boldmath $q$}_{i},\label{Decompose}
	\end{eqnarray}
	where $\mbox{\boldmath $q$}_{i} \in \mathcal{U}^{\perp}$. 
	
	 Let $\mathcal{S}$ be the actual set of codewords that we know only contains codewords from one hemisphere. We rewrite the summation of codewords as
	\begin{eqnarray}
		\mbox{\boldmath $S$} = \sum_{i\in \mathcal{S}}\mbox{\boldmath $x$}_{i} && \!\!\!\!\!\!\!\!\! = \sqrt{nP} \sum_{i\in \mathcal{S}} \mbox{\boldmath $s$}_{i}\label{defS}\\
		&&\!\!\!\!\!\!\!\!\! = \sqrt{nP} \sum_{i\in \mathcal{S}}(t_i \mbox{\boldmath $u$}_{i}+\mbox{\boldmath $q$}_{i})\\
		&& \!\!\!\!\!\!\!\!\! = \underbrace{\left(\sqrt{nP} \sum_{i \in \mathcal{S}}{t}_{i}\right)}_{\triangleq  S_{\parallel}} \mbox{\boldmath $u$}  + \underbrace{\sqrt{nP} \sum_{i\in \mathcal{S}}\mbox{\boldmath $q$}_{i}}_{\triangleq \mbox{\scriptsize \boldmath $S$}_{\perp}}, \label{subcompo}
	\end{eqnarray}
	where (\ref{subcompo}) follows from substituting (\ref{Decompose}) into (\ref{defS}). 
	
	We first focus on parallel component $S_{\parallel}$. Without loss of generality, let $\mathcal{S} = [K_a(n)]$. By Weak Law of Large Numbers (WLLN) as $K_a(n) \rightarrow \infty$, we obtain
	\begin{eqnarray}
		\frac{1}{K_a(n)} \sum_{i=1}^{K_a(n)} t_i \xrightarrow{P} \mu.\label{WLLN1}
	\end{eqnarray}
	Note that throughout this proof, we assume that $\mbox{\boldmath $u$}$ is known and fixed. Therefore, given that $\mbox{\boldmath $s$}_{i}$ are distributed i.i.d. on sphere, the projections $t_i$ are i.i.d. as well, and  we are allowed to use WLLN.
	
	Here, our first goal is to find $\mu$. To this end, we initially introduce the Gaussian Projection Approximation lemma.
	\begin{lemma}(Gaussian Projection Approximation)\label{GPALemma}
		For a fixed unit vector $\mbox{\boldmath $u$}$, and uniformly distributed vector $\mbox{\boldmath $s$}$ on $\mathbb{S}^{n-1}$, let $t= \langle \mbox{\boldmath $u$}, \mbox{\boldmath $s$}\rangle$. Then, we have
		\begin{eqnarray}
			\sqrt{n}t \xrightarrow{d} \mathcal{N}(0,1),
		\end{eqnarray}
		as $n\rightarrow \infty$.
	\end{lemma}
	We know that $K_a(n)$ points are chosen from one hemisphere, so their projection $t_i$ on the true axis of hemisphere $\mbox{\boldmath $u$}$ cannot be negative. Hence, by Lemma~\ref{GPALemma}, it yields
	\begin{eqnarray}
		\mu  = \mathbb{E}[t | t>0] = \sqrt{\frac{2}{\pi n}}.\label{mu1}
	\end{eqnarray}
	By combining the results from (\ref{WLLN1}) and (\ref{mu1}), we have
	\begin{eqnarray}
		\frac{S_{\parallel}}{n}= \frac{K_a(n) \sqrt{nP} \frac{1}{K_a(n)}\sum_{i=1}^{K_a(n)}t_i}{n} \xrightarrow{P} \beta \sqrt{\frac{2P}{\pi}}.\label{SII}
	\end{eqnarray}
	
	Similarly, since we are limited to a hemisphere, as $K_a(n) \rightarrow \infty$, by WLLN, we get
	\begin{eqnarray}
		\frac{1}{K_a(n)} \sum_{i=1}^{K_a(n)} \| \mbox{\boldmath $q$}_{i} \| ^2 \xrightarrow{P}  \mathbb{E}[ \| \mbox{\boldmath $q$}  \| ^2 | t > 0].\label{WLLN2}
	\end{eqnarray}
	 For a perpendicular vector $\mbox{\boldmath $q$}$ in $\mathcal{U}^{\perp}$, condition on $t$, we have
	\begin{eqnarray}
		&& \mathbb{E}[\| \mbox{\boldmath $q$}\| ^2 | t] = \mathbb{E}[\| \mbox{\boldmath $s$} - t\mbox{\boldmath $u$}\| ^2 | t] \nonumber\\
		&& = \mathbb{E}[\| \mbox{\boldmath $s$}\| ^2+ t^2 \| \mbox{\boldmath $u$}\| ^2 - 2 \langle  \mbox{\boldmath $s$}, t\mbox{\boldmath $u$}\rangle | t]\\
		&& = 1 +t^2 - 2t^2 = 1-t^2, \label{1-t2}
	\end{eqnarray}
	and by the result of Lemma~\ref{GPALemma}, we obtain 
	\begin{eqnarray}
		\mathbb{E}[ \| \mbox{\boldmath $q$}  \| ^2 | t > 0] = 1 - \mathbb{E}[t^2 | t>0]= 1 - \frac{1}{2n} \rightarrow 1,\label{Expectperp}
	\end{eqnarray}
	as $n \rightarrow \infty$.
	For the perpendicular component $\mbox{\boldmath  $S$}_{\perp}$, we  write
	\begin{eqnarray}
		&& \!\!\!\!\!\!\!\! \frac{\| \mbox{\boldmath $S$}_{\perp}\| ^  2}{n^2}  = \frac{nP}{n^2} \left \| \sum_{i=1}^{K_a(n)} \mbox{\boldmath $q$}_{i}\right \|  ^ 2 = \nonumber\\
		&&\!\!\!\!\!\!\!\! \frac{P}{n}K_a(n) \frac{1}{K_a(n)}\sum_{i=1}^{K_a(n)}\| \mbox{\boldmath $q$}_{i}\| ^2 + \frac{2P}{n} \!\!\! \sum_{i<j\in [K_a(n)]}\!\!\!\!\!\!\!\! \langle \mbox{\boldmath $q$}_{i}, \mbox{\boldmath $q$}_{j} \rangle. \label{Perpsquared}
	\end{eqnarray}
	According to (\ref{WLLN2}) and (\ref{Expectperp}), we know that the first term in (\ref{Perpsquared}) converges in probability to $P\beta$ as $K_a(n), n \rightarrow \infty$. For a perpendicular vector $\mbox{\boldmath $q$}_{i}$, condition on $t_i$, the unit vectors $\mbox{\boldmath $q$}_{i}/\| \mbox{\boldmath $q$}_{i}\|$ for $i\in [K_a(n)]$ are i.i.d. uniformly distributed on the unit sphere in $\mathcal{U}^{\perp}$. As proved in  \cite{Gorban2018}, \cite{Cai2013}, for any pairs of uniformly distributed vectors on sphere, we have $\langle \mbox{\boldmath $q$}_{i}, \mbox{\boldmath $q$}_{j}\rangle/n \xrightarrow{P} 0$. Hence, the second term in (\ref{Perpsquared}) converges to $0$ in probability. Combining the results, we have
	\begin{eqnarray}
		\frac{\| \mbox{\boldmath $S$}_{\perp}\| ^  2}{n^2} \xrightarrow{P} P\beta. \label{SI_}
	\end{eqnarray}
	Substitute the decomposition version of $\mbox{\boldmath $S$}$ into channel output
	\begin{eqnarray}
		\mbox{\boldmath $Y$} = \mbox{\boldmath $S$}+ \mbox{\boldmath $Z$} = S_{\parallel} \mbox{\boldmath $u$} + \mbox{\boldmath $S$}_{\perp} + \mbox{\boldmath $Z$}, \label{DecompY}
	\end{eqnarray}
	from which
	\begin{eqnarray}\label{Y}
		&& \!\!\!\! \frac{\| \mbox{\boldmath $Y$} \| ^ 2 }{n^2}= \frac{S_{\parallel}^2}{n^2} + \frac{\left \| \mbox{\boldmath $S$}_{\perp} \right \| ^2}{n^2} + \frac{\| \mbox{\boldmath $Z$}\|^2}{n^2} + \frac{2}{n^2} S_{\parallel}\langle \mbox{\boldmath $Z$}, \mbox{\boldmath $u$}\rangle \\
		&& \ \ \ \ \ \ \ \ \ \ + \frac{2}{n^2} \langle \mbox{\boldmath $S$}_{\perp},\mbox{\boldmath $Z$}\rangle + \frac{2S_{\parallel}}{n^2} \underbrace{\langle \mbox{\boldmath $u$}, \mbox{\boldmath $S$}_{\perp}\rangle}_{=0}.
	\end{eqnarray}
	We proved that the first term $S_{\parallel}^2/n^2$ converges in probability to $2P\beta^2/\pi$, and the second term $\left \| \mbox{\boldmath $S$}_{\perp} \right \| ^2/n^2$ converges in probability to $P\beta$. It remains to discuss the remaining terms in (\ref{Y}). By WLLN for i.i.d. $\chi_1^2$ (Chi-squared) distributed random variables $Z_i^2$ for $i\in [n]$, we have
	\begin{eqnarray}
		\frac{\| \mbox{\boldmath $Z$}\| ^2}{n} \xrightarrow{P} 1, \label{|Z|^2}
	\end{eqnarray}
			which results into
	\begin{eqnarray}
		\frac{\| \mbox{\boldmath $Z$}\| ^2}{n^2} \xrightarrow{P} 0.
	\end{eqnarray}
	Given that $\| \mbox{\boldmath $u$} \| =1$ and $\mbox{\boldmath $Z$} \sim \mathcal{N}(0,\mathbf{I}_{n})$, the inner product $\langle \mbox{\boldmath $Z$}, \mbox{\boldmath $u$}\rangle$ has standard normal distribution $\mathcal{N}(0,1)$. Hence, as $n\rightarrow \infty$, we have
	\begin{eqnarray}
		\mathbb{P}\left[\left | \frac{\langle \mbox{\boldmath $Z$}, \mbox{\boldmath $u$}\rangle}{n}\right | > \epsilon \right] < \frac{1}{n^2 \epsilon^2} \rightarrow 0, \label{ConvProb1}
	\end{eqnarray}
	for all $\epsilon >0$, where (\ref{ConvProb1}) follows from Chebyshev inequality. By definition of convergence in probability and from (\ref{ConvProb1}), we obtain
	\begin{eqnarray}
		\frac{\langle \mbox{\boldmath $Z$}, \mbox{\boldmath $u$}\rangle}{n} \xrightarrow{P} 0. \label{<Z,u>}
	\end{eqnarray}
	According to (\ref{SII}) and (\ref{<Z,u>}), using Slutsky's lemma, for the fourth term in (\ref{Y}), we have
	\begin{eqnarray}
		 2 \frac{S_{\parallel}}{n} \frac{\langle \mbox{\boldmath $Z$}, \mbox{\boldmath $u$}\rangle}{n} \xrightarrow{P} 0. \label{SII<Z,u>}
	\end{eqnarray}
	
	Conditional on $\mbox{\boldmath $S$}_{\perp}$, we know that 
	\begin{eqnarray}
		\left \langle \frac{\mbox{\boldmath $S$}_{\perp}}{n^2}, \mbox{\boldmath $Z$} \right \rangle \sim \mathcal{N}\left(0,\frac{\| \mbox{\boldmath $S$}_{\perp}\| ^{2}}{n^4} \right).
	\end{eqnarray}
	Hence, by Chebyshev inequality for any $\epsilon >0$
	\begin{eqnarray}
		&& \!\!\!\!\!\!\! \mathbb{P}\left[ \left | \left \langle \frac{\mbox{\boldmath $S$}_{\perp}}{n^2}, \mbox{\boldmath $Z$}\right \rangle\right| \geq \epsilon \right] =  \nonumber \\
		&& \!\!\!\!\!\!\! \mathbb{E}\left[\mathbb{P}\left[ \left | \left \langle \frac{\mbox{\boldmath $S$}_{\perp}}{n^2}, \mbox{\boldmath $Z$}\right \rangle\right | \geq \epsilon \middle \vert \mbox{\boldmath $S$}_{\perp}\right] \right] \leq \frac{\mathbb{E}\left[\left \| \mbox{\boldmath $S$}_{\perp}\right \| ^2\right]}{n^4 \epsilon ^2}.
	\end{eqnarray}
	Our goal is to prove that $\mathbb{E}\left[\| \mbox{\boldmath  $S$}_{\perp}\| ^2\right]/n^4 \rightarrow 0$ as $n, K_a(n) \rightarrow \infty$, which results into $\langle \mbox{\boldmath $S$}_{\perp}, \mbox{\boldmath $Z$} \rangle /n$ converges to $0$ in probability. To this end, by definition of $\mbox{\boldmath $S$}_{\perp}$, we have
	\begin{eqnarray}
		&& \!\!\!\!\!\!\!\!\!\!\! \frac{\mathbb{E}\left[\| \mbox{\boldmath  $S$}_{\perp} \| ^2 \right]}{n^4} = \frac{nP}{n^4}\left  ( \sum_{i = 1}^{K_a(n)} \mathbb{E}\left[\| \mbox{\boldmath $q$}_{i} \| ^2 \right] + 2 \!\!\!\!\!\! \sum_{i<j \in [K_a(n)]} \!\!\!\!\!\! \mathbb{E}\left[\langle \mbox{\boldmath  $q$}_{i}, \mbox{\boldmath $q$}_{j} \rangle\right]\right )\nonumber \\
		&& \!\!\!\!\!\!\!\!\!\! = \frac{P}{n^3}\left(K_a(n) \mathbb{E}\left [\| \mbox{\boldmath $q$}\| ^2 \middle \vert t>0\right ] +  2 \!\!\!\!\!\!\sum_{i<j\in [K_a(n)]}\!\!\!\!\!\! \left \langle \mathbb{E}[\mbox{\boldmath $q$}_{i}], \mathbb{E}[\mbox{\boldmath $q$}_{j}] \right \rangle  \right)\label{middle_step1}\\
		&& \!\!\!\!\!\!\!\!\!\! = \frac{P\beta}{n^2}\left(1 - \frac{1}{2n}\right)  \rightarrow 0, \label{final_step1}
	\end{eqnarray}
	where (\ref{middle_step1}) follows from $\mbox{\boldmath $q$}_{i}$ being i.i.d., resulting from i.i.d. uniformly distributed $\mbox{\boldmath $s$}_{i}$ restricted to hemisphere. By discussion in \cite[Section 9.3.2]{DS1999}, for points $\mbox{\boldmath $s$}_{i}$ restricted to hemisphere with axis $\mbox{\boldmath $u$}$, the expectation $\mathbb{E}[\mbox{\boldmath $s$}_{i}]$ must be parallel to $\mbox{\boldmath $u$}$, thus for a $b>0$, 
	\begin{eqnarray}
		\mathbb{E}[\mbox{\boldmath $s$}_{i}]  = b \mbox{\boldmath $u$}.\label{ExpDiru}
	\end{eqnarray}
	Taking expectation from both sides of (\ref{Decompose}) and by (\ref{ExpDiru}), we have
	\begin{eqnarray}
		\mathbb{E}[\mbox{\boldmath $q$}_{i}] &&\!\!\!\!\!\!\!\!\!\! = \mathbb{E}[\mbox{\boldmath $s$}_{i}] - \mathbb{E}[\langle \mbox{\boldmath $s$}_{i}, \mbox{\boldmath $u$}\rangle] \mbox{\boldmath $u$}\\
		&& \!\!\!\!\!\!\!\!\!\! = \mathbb{E}[\mbox{\boldmath $s$}_{i}] - \left  \langle \mathbb{E}[\mbox{\boldmath $s$}_{i}], \mbox{\boldmath $u$}\right \rangle \mbox{\boldmath $u$}\\
		&&  \!\!\!\!\!\!\!\!\!\! = b\mbox{\boldmath $u$}  - b\mbox{\boldmath $u$} = 0,
	\end{eqnarray}
which results into zero cross terms in (\ref{middle_step1}). Taking this into account along with (\ref{Expectperp}) yields (\ref{final_step1}). Hence,
\begin{eqnarray}
	\frac{1}{n}\left \langle \mbox{\boldmath $S$}_{\perp}, \mbox{\boldmath $Z$} \right \rangle \xrightarrow{P} 0.\label{<SI_,Z>}
\end{eqnarray}
	According to (\ref{Y}), combining the results from (\ref{SII}), (\ref{SI_}), (\ref{|Z|^2}), (\ref{SII<Z,u>}), and (\ref{<SI_,Z>}) yields to
	\begin{eqnarray}
		\frac{\| \mbox{\boldmath $Y$}\|}{n} \xrightarrow{P}  \sqrt{\frac{2P}{\pi}\beta ^2+P \beta}. \label{Y/n}
	\end{eqnarray}
	
	Using the decomposed version of $\mbox{\boldmath $Y$}$ in (\ref{DecompY}), we obtain
	\begin{eqnarray}
		\frac{1}{n}\left \langle \mbox{\boldmath $Y$}, \mbox{\boldmath $u$}\right \rangle = \frac{S_{\parallel}}{n} + \frac{\langle \mbox{\boldmath $Z$}, \mbox{\boldmath $u$}\rangle}{n},
    \end{eqnarray}	 
    which results into 
    \begin{eqnarray}
    	 \frac{1}{n}\left \langle \mbox{\boldmath $Y$}, \mbox{\boldmath $u$}\right \rangle \xrightarrow{P} \beta \sqrt{\frac{2P}{\pi}} ,\label{<Y,u>}
    \end{eqnarray}
    by (\ref{SII}) and (\ref{<Z,u>}). Finally, combining the results from (\ref{Y/n}) and (\ref{<Y,u>}), we complete the proof as the following
    \begin{eqnarray}
    	\langle \hat{\mbox{\boldmath $u$}}, \mbox{\boldmath $u$} \rangle = \left \langle \frac{\mbox{\boldmath $Y$}}{\| \mbox{\boldmath $Y$}\|}, \mbox{\boldmath $u$} \right \rangle = \frac{\langle \mbox{\boldmath $Y$},  \mbox{\boldmath $u$}\rangle / n}{\| \mbox{\boldmath $Y$}\| / n} \xrightarrow{P} \sqrt{\frac{2\beta}{2\beta + \pi}}.
    \end{eqnarray}
\end{proof}

\section{Proof of Theorem \ref{THMRetention}} \label{ProofTHMRetention}

\begin{proof}
	We decompose $\hat{\mbox{\boldmath $u$}}$ as
	\begin{eqnarray}
		\hat{\mbox{\boldmath $u$}} = c_n \mbox{\boldmath $u$} + \sqrt{1-c_n^2} \mbox{\boldmath $v$}_{n},\label{Decomp_hatu}
	\end{eqnarray}
	where $\mbox{\boldmath $v$}_{n} \in \mathcal{U}^{\perp}$ and $\| \mbox{\boldmath $v$}_{n}\|  = 1$. Here, $c_n \triangleq \langle \hat{\mbox{\boldmath $u$}}, \mbox{\boldmath $u$}\rangle \xrightarrow{P} c$ as proved in Theorem \ref{ThmConvToc}. By decomposition in (\ref{Decomp_hatu}), we have
	\begin{eqnarray}
		\langle \mbox{\boldmath $x$}_{i}, \hat{\mbox{\boldmath $u$}}\rangle = c_n \langle \mbox{\boldmath $x$}_{i}, \mbox{\boldmath $u$}\rangle + \sqrt{1- c_n^2}\langle  \mbox{\boldmath $x$}_{i}, \mbox{\boldmath $v$}_{n}\rangle.
	\end{eqnarray}
	
	Define the normalized score
\begin{eqnarray} \label{T_n}
	T_n \triangleq \sqrt{n}\left\langle \frac{\mbox{\boldmath $x$}_{i}}{\sqrt{nP}},\hat{\mbox{\boldmath $u$}}\right\rangle = \sqrt{n}\langle \mbox{\boldmath $s$}_{i},\hat{\mbox{\boldmath $u$}}\rangle.
\end{eqnarray}
Then
	\begin{eqnarray}
	p_{\mathrm{ret},n}(\tau_n)= \mathbb{P}[T_n \ge \gamma_n | \text{$i$ is the sent codeword}],
\end{eqnarray}
where $\gamma_n=\sqrt n\,\tau_n$.

\subsubsection{For the Search only Limited to $\hat{\mathcal{H}}$}

	For fixed orthonoraml vectors $(\mbox{\boldmath $u$}, \mbox{\boldmath $v$}_{n})$, and $\mbox{\boldmath $s$}_{i} \sim \mathrm{Unif}\{ \mbox{\boldmath $s$}\in \mathbb{S}^{n-1}: \langle \mbox{\boldmath $s$}, \mbox{\boldmath $u$}\rangle  \geq 0\}$ uniformly distributed points on hemisphere with radius $1$, we claim 
	\begin{eqnarray}
		\left( \sqrt{n} \langle \mbox{\boldmath $s$}_{i}, \mbox{\boldmath $u$}\rangle , \sqrt{n} \langle \mbox{\boldmath $s$}_{i},\mbox{\boldmath $v$}_{n}\rangle \right) \xrightarrow{d} (|N_1|,N_2),
	\end{eqnarray}
	where random variables $N_1$ and $N_2$ are independent with distribution $\mathcal{N}(0,1)$. To prove this claim, without loss of generality, let $\mbox{\boldmath $u$} = \mbox{\boldmath  $e$}_{1}$ and $\mbox{\boldmath $v$}_{n} = \mbox{\boldmath $e$}_{2}$. As discussed earlier, we generate uniformly  distributed points on hemisphere by normalizing points generated i.i.d. by Gaussian distribution as
	\begin{eqnarray}
		\mbox{\boldmath $s$}_{i} = \frac{\mbox{\boldmath $P$}}{\| \mbox{\boldmath $P$}\|}, 
    \end{eqnarray}	  
    where $\mbox{\boldmath $P$} = (P_1,...,P_n)  \sim \mathcal{N}(\mbox{\boldmath $0$},  \mathrm{I}_{n})$. Hence,
    \begin{eqnarray}
    	&& \sqrt{n} \langle \mbox{\boldmath $s$}_{i}, \mbox{\boldmath $u$}\rangle = \sqrt{n} \frac{P_1}{\| \mbox{\boldmath $P$}\|},\\
    	&& \sqrt{n} \langle \mbox{\boldmath $s$}_{i}, \mbox{\boldmath $v$}_{n} \rangle = \sqrt{n} \frac{P_2}{\| \mbox{\boldmath $P$}\|}.
    \end{eqnarray}
	By WLLN, we have
	\begin{eqnarray}
		\frac{\| \mbox{\boldmath $P$}\|}{\sqrt{n}} = \sqrt{\frac{1}{n}\sum_{i=1}^{n}P_i^2} \xrightarrow{P} 1.\label{WLLNP}
	\end{eqnarray}
	According to  (\ref{WLLNP}), by Slutsky's lemma and hemisphere constraint ($P_1>0$), we obtain
	\begin{eqnarray}
		&&\!\!\!\!\!\!\!\!\!\!\!\!\!\!\! (\sqrt{n} \langle \mbox{\boldmath $s$}_{i}, \mbox{\boldmath $u$}\rangle, \sqrt{n} \langle \mbox{\boldmath $s$}_{i}, \mbox{\boldmath $v$}_{n} \rangle)= \nonumber\\
		&&\ \ \ \ \ \ \ \ \ \left(\sqrt{n} \frac{P_1}{\| \mbox{\boldmath $P$}\|}, \sqrt{n} \frac{P_2}{\| \mbox{\boldmath $P$}\|} \right) \xrightarrow{d} (|N_1|,N_2), \label{|N_1|N_2}
	\end{eqnarray}
	which proves the claim. Note that since $P_1$ and $P_2$ are independent, $N_1$ and $N_2$ are independent as well. Now, by definition in (\ref{T_n}) and decomposition in (\ref{Decomp_hatu}), we have
	\begin{eqnarray}
		T_n = c_n\sqrt{n}\left\langle \mbox{\boldmath $s$}_i,\mbox{\boldmath $u$}\right\rangle
+\sqrt{1-c_n^2}\sqrt{n}\left\langle \mbox{\boldmath $s$}_i,\mbox{\boldmath $v$}_n\right\rangle.
	\end{eqnarray}
	
	 Given that $c_n \xrightarrow{P} c$, and the result in (\ref{|N_1|N_2}), once again by Slutsky's lemma, we have
	\begin{eqnarray}
		T_n \xrightarrow{d} W_0 \triangleq c|N_1|+\sqrt{1-c^2}N_2. \label{cN-1}
	\end{eqnarray}

	\begin{lemma}(Portmanteau) \label{Portmanteau}
		For any random variables $X_n$ and $X$, the following statements are equivalent:
		\begin{enumerate}
			\item $X_n \xrightarrow{d} X$.
			\item $\mathbb{P}[X_n \in \mathcal{B}] \rightarrow \mathbb{P}[X \in \mathcal{B}]$, for all Borel sets $\mathcal{B}$ provided that $\mathbb{P}[X\in \delta\mathcal{B}] = 0$, where $\delta \mathcal{B} = \bar{\mathcal{B}} - \mathcal{B}^{\circ}$ is the boundary of $\mathcal{B}$, $\bar{\mathcal{B}}$ is the closure of $\mathcal{B}$, and $\mathcal{B}^{\circ}$ is the interior of $\mathcal{B}$.
		\end{enumerate}
	\end{lemma}
	
	If $\gamma_n=0$ for all $n$, since $W_0$ is a continuous random variable, then
	\begin{eqnarray}
		\mathbb{P}[W_0=0]=0,
	\end{eqnarray}
	which satisfies the condition of Portmanteau lemma. Therefore, 
	\begin{equation}
		p_{\mathrm{ret},n}(0)=\mathbb{P}[T_n\geq 0| \text{$i$ is sent}] \rightarrow \mathbb{P}[W_0\geq 0].
	\end{equation}
	we calculate the simple limit probability, regarding the scenario where we only search over $\hat{\mathcal{H}}$ as the following chain 
	\begin{eqnarray}
		&&\!\!\!\!\!\!\!\!\!\!\!\!\!\!\! \mathbb{P}\left[W_0 \geq 0\right]  = \mathbb{E}\left[\mathbb{P}\left[c|N_1|+\sqrt{1-c^2}N_2 \geq 0 \middle \vert |N_1| \right]\right]\\
		&&  = \mathbb{E}\left[\mathbb{P}\left[N_2 \geq \frac{-c|N_1|}{\sqrt{1-c^2}} \middle \vert |N_1| \right]\right]\\
		&& = \mathbb{E}\left[Q\left(\frac{-c|N_1|}{\sqrt{1-c^2}}\right)\right]\\
		&&= \int_{0}^{\infty} \int_{-\infty}^{\frac{cn_1}{\sqrt{1-c^2}}} \frac{1}{\sqrt{2\pi}}e^{-\frac{z^2}{2}} \mathrm{d}z \frac{2}{\sqrt{2\pi}} e^{-\frac{n_1^2}{2}} \mathrm{d}n_1\label{halfnormal}\\
		&&  = \frac{1}{\pi}\int_{0}^{\infty} \int_{-\frac{\pi}{2}}^{\arctan \frac{c}{\sqrt{1-c^2}}}e^{-\frac{r^2}{2}}r\mathrm{d}r \mathrm{d}\theta \label{polar}\\
		&&  = \frac{1}{\pi}\left(\frac{\pi}{2}+\arctan \frac{c}{\sqrt{1-c^2}}\right)\\
		&& = \frac{1}{2} + \frac{1}{\pi}\arcsin c,
	\end{eqnarray}
	where (\ref{halfnormal}) follows from $|N_1|$ being half-normally distributed, and (\ref{polar}) is by writing the integral in polar coordinates.

\subsubsection{For the Search on Sequence of Spherical Caps $\{\hat{\mathcal{H}}_{n}\}$}	
	We follow the same approach as for $\hat{\mathcal{H}}$. Here, again for fixed $(\mbox{\boldmath $u$}, \mbox{\boldmath $v$}_{n})$, the normalized codewords $\mbox{\boldmath $s$}_{i}$ are uniformly distributed on $\{ \mbox{\boldmath $s$}\in \mathbb{S}^{n-1}: \langle \mbox{\boldmath $s$},\mbox{\boldmath $u$}\rangle \geq \tau_n \}$. 
	According to our choice of threshold
\begin{eqnarray}
	\tau_n=-\frac{a_n}{\sqrt{n}},
\end{eqnarray}
where $a_n\to\infty$ and $a_n=o(\sqrt{n})$. Then $\tau_n\to 0^{-}$, so the search region is only slightly larger than the hemisphere, but
\begin{eqnarray}
	\gamma_n=\sqrt{n}\tau_n=-a_n\to -\infty.
\end{eqnarray}
Since the bound $\gamma_n$ diverges negatively, the total variation distance between the conditional $(P_1>0)$ and unconditional projection laws vanishes as $n\rightarrow \infty$. Thus, following the same reasoning, we have 
\begin{eqnarray}
	\left(\sqrt{n}\langle \mbox{\boldmath $s$}_{i},\mbox{\boldmath $u$}\rangle, \sqrt{n}\langle \mbox{\boldmath $s$}_{i}, \mbox{\boldmath $v$}_{n}\rangle\right) \rightarrow (N_1,N_2). 
\end{eqnarray}
	Therefore,
	\begin{eqnarray}
		T_n \xrightarrow{d} W\triangleq cN_1+\sqrt{1-c^2}N_2.
	\end{eqnarray}
	By Portmanteau lemma,
	\begin{eqnarray}
		p_{\mathrm{ret},n}(\tau_n)=\mathbb{P}[T_n\ge \gamma_n| \text{$i$ is sent}] \rightarrow \mathbb{P}[W\ge \gamma].
	\end{eqnarray}

\begin{thm}(Prohorov's Theorem) \cite{vander}
	If the sequence of random variables $\{ X_n\}$ converges in distribution to $X$, then $\{X_n\}$ is uniformly tight or $X_n = O_{P}(1)$.
\end{thm}

Since $T_n \xrightarrow{d} W$, by Prohorov's Theorem, we have $T_n=O_P(1)$. Given that $\gamma_n \rightarrow -\infty$, it yields
\begin{eqnarray}
	\mathbb{P}[T_n\ge \gamma_n]\to 1,
\end{eqnarray}
resulting into
\begin{eqnarray}
p_{\mathrm{ret},n}(\tau_n)\to 1.
\end{eqnarray}

\end{proof}

\section{Proof of Theorem \ref{H_nK_a}} \label{ProofH_nK_a}

\begin{proof}

Fix $n$. We know that $\hat{\mathcal{H}}_{n}$ is larger than $\hat{\mathcal{H}}$, and hence

\begin{eqnarray}
	\mathbb{P}\left[| \hat{\mathcal{H}}_{n}| \geq K_a(n)\right]> \mathbb{P}\left[|\hat{\mathcal{H}}| \geq K_a(n)\right],
\end{eqnarray}
which results into
	\begin{eqnarray}
		\mathbb{P}\left[| \hat{\mathcal{H}}_{n}| < K_a(n)\right] < \mathbb{P}\left[| \hat{\mathcal{H}}| < K_a(n)\right]. \label{H_nH}
    \end{eqnarray}	 
     Our goal in this proof is then shifted to $\mathbb{P}\left[| \hat{\mathcal{H}}| < K_a(n)\right] \rightarrow 0$. W.L.O.G, we assume $\mathcal{S}=[K_a(n)]$. We split the cardinality $| \hat{\mathcal{H}}|$ into 
	\begin{eqnarray}
		| \hat{\mathcal{H}} | = \underbrace{\!\! \sum_{i=1}^{K_a(n)} \mathbf{1}\{\langle \mbox{\boldmath $s$}_{i}, \hat{\mbox{\boldmath $u$}}\rangle \geq 0 \}}_{\triangleq H_{\mathrm{true}}^{(n)}} +\!\! \underbrace{\sum_{j=K_a(n)+1}^{M_n}\!\!\!\!\! \mathbf{1}\{\langle \mbox{\boldmath $s$}_{j},\hat{\mbox{\boldmath $u$}}\rangle \geq 0 \}}_{\triangleq H_{\mathrm{other}}^{(n)}}.\label{CarH1}
	\end{eqnarray}
	A simple observation is that conditional on $\mbox{\boldmath $Y$}$ (or 
$\hat{\mbox{\boldmath $u$}}$), the remaining $M_n-K_a(n)$ points
$\mbox{\boldmath $s$}_{K_a(n)+1},\ldots,\mbox{\boldmath $s$}_{M_n}$ are i.i.d. uniform on
$\mathbb S^{n-1}$ and independent of $\hat{\mbox{\boldmath $u$}}$ (or $\mbox{\boldmath $Y$}$). Therefore, conditional on
$\mbox{\boldmath $Y$}$, the count $H^{(n)}_{\mathrm{other}}$ has distribution $H^{(n)}_{\mathrm{other}} \mid \mbox{\boldmath $Y$} \sim \mathrm{Bin}\!\left(M_n-K_a(n),\,p_n\right)$, where $p_n = \mathbb P\!\left[\sqrt{n}\left\langle \mbox{\boldmath $s$}_{j},\hat{\mbox{\boldmath $u$}}\right\rangle \ge 0\right]$.

By rotational symmetry, we can rotate the direction of $\hat{\mbox{\boldmath $u$}}$ to $\mbox{\boldmath $e$}_{1}$ without changing the distribution of $\sqrt{n}\langle \mbox{\boldmath $s$}_{j}, \hat{\mbox{\boldmath $u$}}\rangle$ \cite{DS1999}, which results into $\sqrt{n}\langle \mbox{\boldmath $s$}_{j}, \hat{\mbox{\boldmath $u$}}\rangle \xrightarrow{d}\mathcal{N}(0,1)$. Hence, by Portmanueau lemma for all $j\in \{K_a(n)+1,...,M_n\}$, it yields
\begin{equation}
	p_n= \mathbb P\!\left[\sqrt{n}\langle\mbox{\boldmath $s$}_{j},	\hat{\mbox{\boldmath $u$}}\rangle \ge 0\right]
\longrightarrow \mathbb P[\mathcal{N}(0,1) \geq 0]=\frac{1}{2}.
\label{pnlimit}
\end{equation}

\textit{Remark 1}. We previously established the limit $p_n \rightarrow 1/2$ by exploiting the rotational symmetry of the uniform distribution on the sphere. Here, we adopt a more rigorous and general approach that can be extended to distributions that are not necessarily rotationally symmetric.
By substituting $\hat{\mbox{\boldmath $u$}} = \mbox{\boldmath $Y$}/  \| \mbox{\boldmath $Y$} \|$, we have
\begin{eqnarray}
		&& \mathbb{P}\left[\langle \mbox{\boldmath $s$}_{i}, \hat{\mbox{\boldmath $u$}}\rangle \geq 0 \middle \vert \text{$i$ is not sent}\right]\nonumber \\
		&& =  \mathbb{P}\left[\left \langle \mbox{\boldmath $s$}_{i}, \sqrt{nP} \sum_{\substack{j=1 \\ j\neq i}}^{K_a(n)}\mbox{\boldmath $s$}_{j} + \mbox{\boldmath $Z$} \right \rangle \geq 0 \right]\label{|Y|Removed}\\
		&& = \mathbb{P}\left[ \sum_{\substack{j=1 \\ j\neq  i}}^{K_a(n)} \langle \mbox{\boldmath $s$}_{i}, \mbox{\boldmath $s$}_{j}\rangle + \left \langle \mbox{\boldmath $s$}_{i}, \frac{\mbox{\boldmath $Z$}}{\sqrt{nP}}\right \rangle \geq 0\right],\label{SumTerm}
	\end{eqnarray}
	where (\ref{|Y|Removed}) is by $\| \mbox{\boldmath $Y$}\| >0$. 
	
	\begin{defn}(Triangular Array of Random Variables)
		Suppose that for each $n$, random variables 
		\begin{eqnarray}
			X_{n,1},...,X_{n,r_n}
		\end{eqnarray}
		are independent; the probability space for the sequence may change with $n$. Such a collection is called \emph{triangular array} of random variables. \cite{Bilingsly2012}
	\end{defn}
	
	Condition on $\mbox{\boldmath $s$}_{i}$, let define
	\begin{eqnarray}
		N_{n,j} \triangleq \langle  \mbox{\boldmath $s$}_{i}, \mbox{\boldmath $s$}_{j} \rangle.
	\end{eqnarray}
	Random variables $\{N_{n,j}\}_{j=1}^{K_a(n)}$ given $\mbox{\boldmath $s$}_{i}$ are independent, because $\mbox{\boldmath $s$}_{j}$ are generated independently. Due to symmetry, the distribution of $N_{n,j}$ given $\mbox{\boldmath $s$}_{i}$ is the same as the first coordinate of a uniform vector on $\mathbb{S}^{n-1}$ with PDF \cite{DS1999}
	\begin{eqnarray}
		f_{N_{n,j} | \mbox{\scriptsize \boldmath $s$}_{i}}(t) = \frac{1}{\sqrt{\pi}}\frac{\Gamma\left(\frac{n}{2}\right)}{\Gamma\left(\frac{n-1}{2}\right)}(1-t^2)^{\frac{n-3}{2}} \ \mathrm{for} \ |t|\leq 1,  \label{DisNj} 
	\end{eqnarray}
		where $\Gamma(.)$ is the Gamma function. It is clear that the distribution of $N_{n,j}$ given $\mbox{\boldmath $s$}_{i}$ depends on $n$, with mean and variance
	\begin{eqnarray}
		\mathbb{E}\left[N_{n,j}  \middle \vert \mbox{\boldmath $s$}_{i}\right] = 0, \ \ \ \mathrm{Var}\left(N_{n,j} \middle \vert \mbox{\boldmath $s$}_{i}\right) = \frac{1}{n}.
	\end{eqnarray}
Therefore, the terms $\{N_{n,j}\}_{j=1}^{K_a(n)}$ forms a \emph{triangular array}, as both the distribution of its components and the number of terms depend on $n$. While intuitively the summation term in (\ref{SumTerm}), conditional on $\mbox{\boldmath $s$}_{i}$, should converge to a Gaussian distribution, the classical Central Limit Theorem (CLT) does not apply because it requires i.i.d. random variables whose distribution is fixed. To justify the results rigorously, we instead employ Lindeberg-Feller CLT. We first recall Lindeberg-Feller CLT.

\begin{thm}(Lindeberg-Feller CLT) \cite{Bilingsly2012}
 	Suppose $\{X_{n,j}\}$ is a triangular array with $\mathbb{E}[X_{n,j}] = 0$, $\mathbb{E}[X_{n,j} ^2] = \sigma_{n,j}^2$, and $s_n^2 \triangleq \sum_{j=1}^{r_n}\sigma_{n,j} ^2$. If the Lindeberg condition holds
	\begin{eqnarray}
		\frac{1}{s_n^2} \sum_{j=1}^{r_n}\mathbb{E}\left[X_{n,j}^2 \mathbf{1}\{ |X_{n,j}| > \epsilon s_n \} \right] \rightarrow 0,  \label{LBCon}
	\end{eqnarray}
	for all $\epsilon >0$ as $n\rightarrow \infty$, then 
	\begin{eqnarray}
		\frac{1}{s_n} \sum_{j=1}^{r_n} X_{n,j}  \xrightarrow{d} \mathcal{N}(0,1). \nonumber
	\end{eqnarray}
\end{thm} 
To fit our problem within the framework of Lindeberg-Feller CLT, we verify that the Lindeberg condition holds for $\{N_{n,j}\}$. In our problem, we have
	\begin{eqnarray}
		s_n^2  = \sum_{\substack{j=1 \\ j\neq i}}^{K_a(n)} \mathrm{Var}(N_j | \mbox{\boldmath $s$}_{i}) = \frac{K_a(n)}{n} \rightarrow \beta. \label{sn2}
	\end{eqnarray}
For a fixed $\epsilon>0$, we focus on the argument of the sum in Lindeberg condition (\ref{LBCon}) for $\{N_{n,j}\}$ as follows
	\begin{eqnarray}
		&& \mathbb{E}\left[N_{n,j}^2 \mathbf{1}\{ |N_{n,j}| >\epsilon s_n \} \middle \vert \mbox{\boldmath $s$}_{i}\right] \stackrel{(a)}{\leq} \mathbb{P}\left[|N_{n,j}|>\epsilon s_n  \middle \vert \mbox{\boldmath $s$}_{i}\right] \nonumber\\
		&& \stackrel{(b)}{=} \frac{2}{\sqrt{\pi}} \frac{\Gamma\left (\frac{n}{2}\right)}{\Gamma\left(\frac{n-1}{2}\right)}\int_{\epsilon s_n}^{1}(1-t^2)^{\frac{n-3}{2}} \mathrm{d}t \\
		&&  < \frac{2}{\sqrt{\pi}} \frac{\Gamma\left(\frac{n}{2}\right)}{\Gamma\left(\frac{n-1}{2}\right)}\int_{\epsilon s_n}^{1} e^{-\frac{n-3}{2}t^2} \mathrm{d}t \label{log}\\
		&&  < \frac{\Gamma\left(\frac{n}{2}\right)}{\Gamma\left(\frac{n-1}{2}\right)}\sqrt{\frac{2}{n-3}} e^{-\frac{n-3}{2}\epsilon ^2 s_n^2},\label{GaussianTail}
	\end{eqnarray}
where $(a)$ is by $|N_{n,j}| \leq 1$, $(b)$ follows from the distribution of $N_{n,j}$ given $\mbox{\boldmath $s$}_{i}$ in (\ref{DisNj}), the upper bound in (\ref{log}) holds by $\log(1-x) \leq -x$ for $0\leq x <1$, and the bound in (\ref{GaussianTail}) follows from the Gaussian tail 
	\begin{eqnarray}
		\int_{a}^{\infty}e^{-\alpha t^2} \mathrm{d}t \leq \frac{1}{2}\sqrt{\frac{\pi}{\alpha}}e^{-\alpha  a^2}.
	\end{eqnarray}
Substituting (\ref{GaussianTail}) into Lindeberg condition (\ref{LBCon}) yields
	\begin{eqnarray}
		&&\frac{1}{s_n^2} \sum_{\substack{  j=1 \\ j\neq i}}^{K_a(n)} \mathbb{E}\left[N_{n,j}^2 \mathbf{1}\{ |N_{n,j}| > \epsilon s_n \}  \middle \vert \mbox{\boldmath $s$}_{i}\right]\nonumber \\
		&& < \frac{K_a(n)}{s_n^2}\frac{\Gamma\left(\frac{n}{2}\right)}{\Gamma\left(\frac{n-1}{2}\right)}\sqrt{\frac{2}{n-3}}e^{-\frac{n-3}{2}\epsilon ^2 s_n^2}\\
		&&  \sim n \sqrt{\frac{n}{n-3}}e^{-\frac{n-3}{2}\epsilon ^2 \beta} \rightarrow 0,\label{ApproxGamma}
	\end{eqnarray}
as $n\rightarrow \infty$, where (\ref{ApproxGamma}) follows from the asymptotic behavior of 
	\begin{eqnarray}
		\frac{\Gamma\left(\frac{n}{2}\right)}{\Gamma\left(\frac{n-1}{2}\right)} \sim \sqrt{\frac{n}{2}},
	\end{eqnarray}
for large $n$ and substituting $s_n^2$ given in (\ref{sn2}). Now that we proved ${N_j}$ satisfies Lindeberg condition, by Lindeberg-Feller CLT, we obtain
	\begin{eqnarray}
		\frac{1}{s_n} \sum_{\substack{j=1 \\ j\neq i}}^{K_a(n)}N_{n,j} \xrightarrow{d} \mathcal{N}(0,1),
	\end{eqnarray}
or 
	\begin{eqnarray}
		\sum_{\substack{j=1 \\ j\neq i}}^{K_a(n)}N_{n,j} \xrightarrow{d} \mathcal{N}(0,\beta). \label{NormalBeta}
	\end{eqnarray}

Since $\mbox{\boldmath $Z$} \sim \mathcal{N}(0,\mathrm{I}_{n})$, given $\mbox{\boldmath $s$}_{i}$, random variable $\langle \mbox{\boldmath $s$}_{i}, \mbox{\boldmath $Z$} / \sqrt{nP}\rangle$ has distribution $\mathcal{N}\left(0, \frac{1}{nP}\right)$. 
Hence, for all $\epsilon >0$, by Chebyshev inequality, we have
	\begin{eqnarray}
		\mathbb{P}\left[\left | \left\langle \mbox{\boldmath $s$}_{i}, \frac{\mbox{\boldmath $Z$}}{\sqrt{nP}}\right \rangle\right | > \epsilon\right] \leq \frac{1}{nP\epsilon^2} \rightarrow 0,
	\end{eqnarray}
	 as $n \rightarrow \infty$, resulting into 
	 \begin{eqnarray}
	 	\left\langle \mbox{\boldmath $s$}_{i}, \frac{\mbox{\boldmath $Z$}}{\sqrt{nP}}\right \rangle \xrightarrow{P} 0. \label{InProbtoZero}
	 \end{eqnarray}
	 Now, from results in (\ref{NormalBeta}) and (\ref{InProbtoZero}), by Slutsky's lemma, it yields that
	\begin{eqnarray}
		\sum_{\substack{j=1  \\ j\neq i}}^{K_a(n)} \langle \mbox{\boldmath $s$}_{i}, \mbox{\boldmath $s$}_{j}\rangle + \left \langle \mbox{\boldmath $s$}_{i},\frac{\mbox{\boldmath $Z$}}{\sqrt{nP}}\right \rangle  \xrightarrow{d} \mathcal{N}(0,\beta).
	\end{eqnarray}
	Again, by Portmanteau Lemma (Lemma~\ref{Portmanteau}), as $n\rightarrow \infty$, we have
	\begin{eqnarray}
		&& \mathbb{P}\left[\sum_{\substack{j=1  \\  j\neq i}}^{K_a(n)}\langle \mbox{\boldmath $s$}_{i}, \mbox{\boldmath $s$}_{j}\rangle + \left \langle  \mbox{\boldmath $s$}_{i}, \frac{\mbox{\boldmath $Z$}}{\sqrt{nP}}\right\rangle \geq  0 \middle \vert \mbox{\boldmath $s$}_{i}\right] \nonumber\\
		&& \ \ \ \ \ \ \ \ \ \ \ \ \ \ \ \ \ \longrightarrow\mathbb{P}\left[\mathcal{N}(0,\beta)   \geq 0\right] =  \frac{1}{2}.\label{1/2Limit}
	\end{eqnarray}

Taking expectation from both sides of \eqref{CarH1}, we have
\begin{eqnarray}
	\mathbb E\!\left[\left|\hat{\mathcal H}\right|\right]
	&& \!\!\!\!\!\!\!\!\!\!\!\! = \mathbb E\!	\left[\sum_{i=1}^{K_a(n)} \mathbf{1}\!\left\{\left\langle \mbox{\boldmath $s$}_i,\hat{\mbox{\boldmath $u$}}\right\rangle \ge 0\right\}\right]
	+ \mathbb E\!\left[\mathbb E\!\left[H^{(n)}_	{\mathrm{other}} \middle| \mbox{\boldmath $Y$}\right]\right] \nonumber\\
	&& \!\!\!\!\!\!\!\!\!\!\!\!\!\!\!\!\!\!\!\!\!\!\!\!\!\!	\!\!\!\!\!\!\!\!\!\! = K_a(n)\,\mathbb P\!\left[\left\langle \mbox{\boldmath $s$}_i,	\hat{\mbox{\boldmath $u$}}\right\rangle \ge 0 \,\middle|\, \mbox{\boldmath $s$}_i \ \text{is sent} \right]
\! + \! \big(M_n-K_a(n)\big)p_n. \label{EXPH_cap1}
\end{eqnarray}
As proved in Theorem \ref{THMRetention}, we know that
\begin{equation}
p_{\mathrm{ret},n}(0)=\mathbb P\!\left[\left\langle \mbox{\boldmath $s$}_i,\hat{\mbox{\boldmath $u$}}\right\rangle \ge 0 \,\middle|\, \mbox{\boldmath $s$}_i \ \text{is sent} \right]
\rightarrow \frac{1}{2}+\frac{1}{\pi}\arcsin c,
\end{equation}
as $n \rightarrow \infty$. 
Hence, for every $\epsilon>0$, there exists $n_0>0$ such that for all $n\ge n_0$,
\begin{eqnarray}
	&& \!\!\!\!\!\!\!\!\!\!\!\!\!\!\!\!\!\!\!\!\!\!\! \frac{1}{2}+\frac{1}{\pi}\arcsin c - \epsilon < p_{\mathrm{ret},n}(0) <\frac{1}{2}+\frac{1}{\pi}\arcsin c + \epsilon , \label{DefLimit_cap2}\\
	&& \ \ \ \ \ \ \ \frac{1}{2}-\epsilon < p_n <\frac{1}{2}+\epsilon.
\label{DefLimit_cap1}
\end{eqnarray}
By \eqref{EXPH_cap1}, \eqref{DefLimit_cap2}, \eqref{DefLimit_cap1}, and using
$M_n= n^d$ and $K_a(n)=\beta n$, for all large $n$, we obtain

\begin{eqnarray}
	 &&\!\!\!\!\!\!\!\!\!\!\!\!\!\!\!\!\!\!\!\!\! \mathbb E\!\left[\left|\hat{\mathcal H}\right|\right] \geq (\beta n) \left(\frac{1}{2}+\frac{1}{\pi}\arcsin c -\epsilon \right) + \nonumber \\
	  && \ \ \ \ \ \ \ \ \ \ \ \ \ \ \ \ \ \ (n^d-\beta n)\left(\frac{1}{2}-\epsilon\right),
\label{LowerboundonEH_cap1}
\end{eqnarray} 
Since $d>2$, the dominant term is $(1/2-\epsilon)n^d$. Hence, there exists a constant $c_0>0$ such that, for all large $n$, 
\begin{eqnarray}
	\mathbb{E}\left[|\hat{\mathcal{H}}|\right]\geq c_0 n^d.
\end{eqnarray}
Because $K_a(n)=n\beta$, we sharpen this to
\begin{eqnarray}
	\mathbb{E}\left[| \hat{\mathcal{H}}|\right]-K_a(n) \geq c^{\prime} n^d
\end{eqnarray}
for some constant $c^{\prime}>0$ and all sufficiently large $n$.

	We start upper bounding $\mathbb{P}\left[|\hat{\mathcal{H}}| < K_a(n) \right]$ as the following
	\begin{eqnarray}
		&& \!\!\!\!\!\!\!\!\!\! \mathbb{P}\left[|\hat{\mathcal{H}}| < K_a(n) \right] = \mathbb{P}\left[| \hat{\mathcal{H}}| -  \mathbb{E}\left[|\hat{\mathcal{H}}|\right]\! <\! K_a(n) -\mathbb{E}\left[|\hat{\mathcal{H}}|\right]  \right] \nonumber\\
		&& \!\!\!\!\!\!\!\!\!\! \leq \mathbb{P}\left[| \hat{\mathcal{H}}| -  \mathbb{E}\left[|\hat{\mathcal{H}}|\right] < -c^{\prime}n^d \right]\nonumber\\
		&& \!\!\!\!\!\!\!\!\!\! \leq \mathbb{P}\left[H_{\mathrm{true}}^{(n)} - \mathbb{E}\left[H_{\mathrm{true}}^{(n)}\right] < -\frac{c^{\prime}}{2}n^d\right] \label{bound_based_on_split1}\\
		&& \ \ \ \ \ \ \ + \mathbb{P}\left[H_{\mathrm{other}}^{(n)} - \mathbb{E}\left[H_{\mathrm{other}}^{(n)}\right] < -\frac{c^{\prime}}{2}n^d\right],\label{bound_based_on_split}
	\end{eqnarray}
	where (\ref{bound_based_on_split1}) and (\ref{bound_based_on_split}) follow from the split of $|\hat{\mathcal{H}}|$ in (\ref{CarH1}). We next focus on upper bounding the probability terms in (\ref{bound_based_on_split1}) and (\ref{bound_based_on_split}) according to suitable concentration inequalities. 
	\begin{itemize}
		\item \emph{Concentration for $H_{\mathrm{other}}^{(n)}$ (Hoeffding's Inequality)}: Since $H_{\mathrm{other}}^{(n)} | \mbox{\boldmath $Y$} \sim  \mathrm{Bin}\left(M_n-K_a(n), p_n\right)$ such that $p_n \rightarrow 1/2$, by Hoeffding's inequality, we have
		\begin{eqnarray}
			&& \mathbb{E}\left[\mathbb{P}\left[H_{\mathrm{other}}^{(n)} - \mathbb{E}\left[H_{\mathrm{other}}^{(n)}\right] < -\frac{c^{\prime}}{2}e^{nR} \middle \vert \mbox{\boldmath $Y$}\right]\right] \nonumber\\
			&&\leq \exp\left \{-\frac{{c^{\prime}}^2 n^{2d}}{2(M_n-K_a(n))}\right\} \stackrel{(a)}{=} e^{-\tilde{c}_{1}n^d}, \label{HoeffdingBound}
        \end{eqnarray}	
        where $(a)$ follows from assumptions $M_n = n^d$, $K_a(n) = o(n^d)$ and $M_n-K_a(n) \sim n^d$ for large $n$, resulting into a constant $\tilde{c}_{1}>0$.
        \item \emph{Concentration for $H_{\mathrm{true}}$ (McDiarmid's Inequality)}:
        Before proceeding, we recall the bounded difference property along with McDiarmid's inequality.
        \begin{defn}(Bounded Difference Property)
        	A function $f : \mathcal{X}_{1} \times \mathcal{X}_{2} \times ... \times \mathcal{X}_{n} \rightarrow \mathbb{R}$ satisfies bounded difference property, if substituting the value of the $i$-th coordinate $x_i$ changes the value of $f$ by at most $c_i$.
        \end{defn}

        \textbf{McDiarmid's Inequality}: Let $f : \mathcal{X}_{1} \times \mathcal{X}_{2} \times ... \times \mathcal{X}_{n} \rightarrow \mathbb{R}$ be a function satisfying bounded difference property with bounds $c_1,...,c_n$. Consider independent random variables $X_1,...,X_n$, where $X_i \in \mathcal{X}_{i}$ for all $i$. Then $\forall \epsilon > 0$,
        \begin{eqnarray}
        	&& \mathbb{P}\left[f(X_1,...,X_n) - \mathbb{E}\left[f(X_1,...,X_n)\right] \leq -\epsilon \right] \nonumber \\
        	&&  \ \ \ \ \ \ \ \ \ \ \ \ \ \ \ \ \ \ \ \ \ \ \leq \exp\left \{ -\frac{2\epsilon ^2}{\sum_{i=1}^{n}c_i^2}\right\}.\nonumber
        \end{eqnarray}
        As defined in (\ref{CarH1}), we have
        \begin{eqnarray}
        	H_{\mathrm{true}}^{(n)} = \sum_{i=1}^{K_a(n)} \mathbf{1}\{ \langle \mbox{\boldmath $s$}_{i}, \hat{\mbox{\boldmath $u$}}\rangle \geq 0 \}.
        \end{eqnarray}         
        We consider $H_{\mathrm{true}}^{(n)}$ as a function of $K_a(n)$ codewords and Gaussian noise $\mbox{\boldmath $Z$}$. Each codewod and $\mbox{\boldmath $Z$}$ are independent inputs to function $H_{\mathrm{true}}^{(n)}$, and changing any single input can change $H_{\mathrm{true}}^{(n)}$ by at most $1$ ($H_{\mathrm{true}}^{(n)}$ is a summation of indicator functions that only take values $0$ and $1$.). Thus, $H_{\mathrm{true}}^{(n)}$ satisfies the bounded difference property with bounds $c_1,...,c_{K_a(n)+1} \leq 1$. Using McDiarmid's inequality, we get
        \begin{eqnarray}
        	&& \mathbb{P}\left[ H_{\mathrm{true}}^{(n)} - \mathbb{E}\left[H_{\mathrm{true}}^{(n)}\right] \leq -\frac{c^{\prime}n^d}{2}\right] \nonumber\\
        	&& \leq \exp\left\{ -\frac{{c^{\prime}}^2 n^{2d}}{2(K_a(n)+1)}\right\}\stackrel{(b)}{=}e^{-\tilde{c}_{2}n^{2d-1}},\label{McDiarmid}
        \end{eqnarray}
        where $(b)$ follows from ${K_a(n)}/{n} \rightarrow \beta$ for some constant $\tilde{c}_{2}>0$.
	\end{itemize}
	Now, according to (\ref{bound_based_on_split1}) and (\ref{bound_based_on_split}), combining the results from (\ref{HoeffdingBound}) and (\ref{McDiarmid}) yields
	\begin{eqnarray}
		&& \mathbb{P}\left[| \hat{\mathcal{H}}| <  K_a(n) \right] \leq \mathbb{P}\left[| \hat{\mathcal{H}}| - \mathbb{E}\left[| \hat{\mathcal{H}} | \right]  \leq -cn^d \right] \nonumber \\
		&& \ \ \ \ \ \ \ \ \ \ \ \ \ \ \ \ \ \ \ \ \ \leq  e^{-\tilde{c}_{2}n^{2d-1}} + e^{-\tilde{c}_{1}n^d}, \nonumber
	\end{eqnarray}
	which results into
	\begin{eqnarray}
		\mathbb{P}\left[| \hat{\mathcal{H}}| <  K_a(n) \right] \rightarrow 0,
	\end{eqnarray}
	as $n\rightarrow \infty$. Finally, from (\ref{H_nH}), it follows that
	\begin{eqnarray}
		\mathbb{P}\left[| \hat{\mathcal{H}}_{n}| < K_a(n)\right] \rightarrow 0.
	\end{eqnarray}
\end{proof}

\end{document}